\documentclass{article}

\usepackage[utf8]{inputenc}
\usepackage[T1]{fontenc}
\usepackage{hyperref}
\usepackage{url}
\usepackage{booktabs}
\usepackage{amsfonts}
\usepackage{nicefrac}
\usepackage[expansion=false]{microtype}
\usepackage{xcolor}
\usepackage{amsmath}
\usepackage{amssymb}
\usepackage{mathtools}
\usepackage{wrapfig}
\usepackage{amsthm}
\usepackage{multirow}
\usepackage{graphicx}
\usepackage{natbib}
\usepackage[margin=1in]{geometry}
\usepackage{float}
\usepackage{bm}
\usepackage{authblk}
\usepackage{array}
\usepackage[ruled,linesnumbered]{algorithm2e}
\usepackage{threeparttable}
\usepackage{placeins}
\makeatletter
\@ifundefined{c@definition}{\newtheorem{definition}{Definition}}{}
\@ifundefined{c@proposition}{\newtheorem{proposition}{Proposition}}{}
\makeatother

\title{When Does Overlap Help? OSU-Mem and a Cell-Conditional Analysis of Trajectory Memory for LLM Agents}

\author[1]{Mellow Baixuan Chen}
\author[2]{Xiangguo Sun$^{\dagger}$}
\affil[1]{Courant Institute of Mathematical Sciences, New York University, New York, USA}
\affil[2]{Southeast University, Nanjing, China}
\date{}

\begin{document}

\maketitle

\begin{center}
\small
\textsuperscript{$\dagger$}Corresponding author: Xiangguo Sun (\texttt{hsiang-kuo.sun@seu.edu.cn}).
\end{center}

\begin{abstract}
Long-horizon large language model (LLM) agents accumulate interaction trajectories that quickly exceed any practical prompt budget, and existing memory methods either truncate aggressively and lose non-local evidence or retain boilerplate that degrades decision quality. We ask a \emph{mechanism} question rather than claiming a better general-purpose memory system: \emph{when} does organizing trajectory memory into \emph{overlapping semantic units} (OSUs)---groups of related steps in which one step may belong to several units---help retrieval over flat or disjoint alternatives? We instantiate this in \textbf{OSU-Mem}, which retrieves from an overlapping OSU pool via budgeted coarse-to-fine expansion, and show its benefit is conditional: overlapping memory helps when the evidence steps a query needs share tool calls or entities, but hurts when those steps are fully heterogeneous and share neither. On a synthetic benchmark where evidence carries such shared structure by construction, OSU-Mem improves over the strongest baseline as the theory predicts; yet on a concatenated, constructed unaugmented \textbf{$\tau$-bench} setting its aggregate advantage over flat retrieval vanishes. Splitting queries by whether their evidence shares tools and entities shows this near-tie to be an artifact of mixing query types rather than a property of either method, and \textbf{ToolBench}, a controlled probe built to carry shared structure by design, corroborates the same mechanism via an overlap-vs.-disjoint \emph{construction} contrast (under a coverage-guided variant), isolating the construction principle rather than validating the full default system. Because the relevant sharing is cheaply estimable from metadata, the analysis yields a metadata-based heuristic for predicting when overlap is likely to improve retrieval. We deliberately isolate the retrieval layer, assessed by retrieval quality and an LLM-mediated evidence-selection stage. 
Code available at \href{https://anonymous.4open.science/r/OSUMem}{here}.
\end{abstract}

\section{Introduction}
\label{sec:intro}

LLM agents deployed in web- and tool-augmented environments routinely accumulate trajectories of dozens to hundreds of interaction steps \citep{kang2026aconoptimizingcontextcompression, ye2026agentfold}. Treating the full trajectory as prompt context is quickly infeasible, while compression-based approaches, from fixed-window truncation to learned periodic summarization \citep{lu2025scalingllmmultiturnrl}, risk discarding details that may become relevant later. The tension is real: aggressive compression forgets what a future query needs, while conservative retention floods the prompt with boilerplate that degrades decision quality.

Two properties of agent trajectories make existing approaches brittle. First, relevance is \emph{non-local}, since evidence steps needed at decision time may be far apart in trajectory time. Second, trajectories contain \emph{overlapping semantic threads} such as evidence about an entity, progress toward a subgoal, or a recurring tool-use motif, where a single step can participate in multiple threads simultaneously. Methods that assume temporal locality miss non-local evidence, and methods that enforce disjoint partitions confine multi-purpose steps to a single bucket and cannot reach them from other angles.

We instantiate these properties in \textbf{OSU-Mem}, which treats the interaction history itself as long-term memory and organizes it into \emph{overlapping semantic units} (OSUs). Each OSU groups semantically related steps, and OSUs may overlap so that a step participates in multiple threads. On top of this flat OSU pool, OSU-Mem performs budgeted query-adaptive coarse-to-fine expansion: retrieve a small set of coarse OSU representations, then selectively expand the most query-relevant ones into member steps until the budget is exhausted. We isolate the memory retrieval layer as our object of study, proposing OSU-Mem and characterizing a \emph{conditional} mechanism for when its overlap construction helps, rather than claiming it as a universal, general-purpose agent-memory system; closed-loop agent task interaction is out of scope.

\paragraph{A note on headline framing.} On a constructed, concatenated \textbf{unaugmented} $\tau$-bench setting, aggregated across all queries, OSU-Mem is \emph{not} better than Flat Vector (per-cell and aggregate figures in Section~\ref{sec:cross-tab}). Taken at face value, this aggregate would suggest OSU-Mem does not help on real trajectories. Our central argument is that this aggregate is a \emph{cell-mixture artifact}. Throughout, we label each query by whether its ground-truth evidence steps pairwise share a \emph{tool} signature (T+ if so, T$-$ if not) and an \emph{entity} (E+ if so, E$-$ if not), which partitions queries into four cells (T+E+, T+E$-$, T$-$E+, T$-$E$-$) defined formally in Definition~\ref{def:tpep}. The evidence distribution of raw $\tau$-bench is dominated by the T$-$E$-$ cell (evidence sharing neither tool nor entity), where overlap construction is actively harmful, which masks a large and significant advantage on the T+E+ cell (evidence sharing both). Throughout the paper we report both the aggregate and the cell-conditioned breakdown, and our claims are explicitly conditional on cell structure rather than universal. Our evaluation targets evidence retrieval quality and LLM-mediated evidence selection.

\paragraph{Contributions.} We frame OSU-Mem as a probe for a mechanism question---\emph{when does overlapping trajectory memory help retrieval?}---rather than as a general-purpose memory system, and our contributions are correspondingly conditional.
\begin{itemize}
\item We propose OSU-Mem and use a mechanism-driven synthetic benchmark, in which evidence shares an entity by construction (E+ for $100\%$ of queries; pairwise tool-signature sharing holds for $39.3\%$, so the strict T+E+ fraction is $39.3\%$), as a controlled testbed: under these conditions OSU-Mem improves Recall by $+39.9\%$ and Hit@2 (the fraction of queries that retrieve at least two evidence steps) by $+61.5\%$ relative to the strongest baseline at $B{=}256$, confirming the mechanism in the regime the theory predicts (Section~\ref{sec:results}).
\item On trajectories derived from \textbf{$\tau$-bench} (Section~\ref{sec:results}), a $2{\times}2$ evidence-structure cross-tabulation (cross-tab) localizes \emph{where} overlap helps: OSU-Mem outperforms Flat Vector on the T+E+ sub-population (stably across budgets) and is dominated on the T$-$E$-$ sub-population (increasingly so at larger budgets). The raw-cache aggregate near-tie is thus a consequence of cell mixture, not a property of either method. A controlled multi-burst augmentation that lifts all queries into T+E+ recovers the synthetic ordering---a prediction made from the cross-tab before the augmentation was run---and an LLM-mediated evidence-selection stage (Section~\ref{sec:llm-e2e}) finds the T+E+ advantage surviving an LLM filtering step.
\item On \textbf{ToolBench} (Section~\ref{sec:results}) we construct a controlled mechanism probe that is \emph{structurally T+E+ by design}: planted evidence steps are drawn from a single source trajectory and therefore pairwise share a tool signature (via the \texttt{<fname>\_for\_<api>} convention) and argument entities. Overlap construction beats disjoint construction with a monotonic dose-response in the coverage bonus, and disjoint assignment strips a fraction of evidence steps of every retrieval handle while overlap retains all of them. We treat ToolBench as a high-power measurement of the effect \emph{given} the favorable cell, not as evidence of deployment utility on natural ToolBench agents.
\item We argue and empirically support a \textbf{unified mechanism}: the two datasets give complementary evidence---$\tau$-bench supplying the cross-cell heterogeneity and ToolBench isolating the favorable cell---for the same cell-conditional effect across two evidence distributions and two retrieval algorithms. We further identify a cross-source metadata contamination \textbf{deployment pitfall} (fixed by one-line namespace prefixing---a practical lesson for anyone concatenating trajectories from multiple sources) and derive a \textbf{predictive rule}, a metadata-based estimate of when overlap helps \emph{retrieval}, from the $\tau$-bench cross-tab, whose budget-dependent thresholds are given in Appendix~\ref{app:deployment}.
\end{itemize}

\section{Related Work}
\label{sec:related}

\textbf{Memory for long-horizon LLM agents.} Operating-system-inspired virtual-context systems page interaction history in and out of the prompt \citep{packer2024memgptllmsoperatingsystems}. A-MEM organizes long-term memory via dynamically linked notes with memory evolution \citep{xu2025amemagenticmemoryllm}, HiAgent manages in-trial working memory through subgoal-based hierarchical chunking \citep{hu-etal-2025-hiagent}, and Zep models memory as a temporal knowledge graph \citep{rasmussen2025zeptemporalknowledgegraph}. More recent work improves efficiency through semantic compression and multi-view indexing \citep{liu2026simplememefficientlifelongmemory} or learns explicit memory operations via reinforcement learning \citep{yan2026memoryr1enhancinglargelanguage}. These systems establish the value of structured memory beyond naive replay, but typically organize memory into non-overlapping units or retrieve at fixed granularity, which limits reuse when one step is relevant under multiple semantic views.

\textbf{Experiential and reflective memory.} Reflexion \citep{shinn2023reflexionlanguageagentsverbal}, ExpeL \citep{zhao2024expelllmagentsexperiential}, and Voyager \citep{wang2023voyageropenendedembodiedagent} treat trajectories as learning signals via verbal self-reflection, extracted insights, or executable skill libraries. Their retrieval mechanisms are comparatively simple (prompt-append, task-level semantic similarity, embedding lookup) and lack budget-aware granularity control. OSU-Mem addresses a different layer, namely how to \emph{organize and index} raw trajectory steps for budgeted retrieval under non-local relevance.

\textbf{Multi-granularity and overlapping indexing in Retrieval-Augmented Generation (RAG).} Hierarchical retrieval via recursive clustering \citep{sarthi2024raptorrecursiveabstractiveprocessing} and tree-structured navigation \citep{chen2023walkingmemorymazecontext} have been explored in long-document RAG, and coarse-to-fine selection is a standard RAG theme \citep{gao2024retrievalaugmentedgenerationlargelanguage}. Practical RAG systems also support limited forms of overlap (parent-child chunks, sentence windows) \citep{gao2024retrievalaugmentedgenerationlargelanguage}. Our contribution is specifically at the intersection of \emph{budgeted retrieval from agent trajectories where semantic threads overlap and the same step must remain accessible under multiple contexts}, with an explicit cell-level characterization of when overlap pays off.

\textbf{Scope of baselines.} Because OSU-Mem is a contribution at the \emph{retrieval layer}, our comparisons focus on \emph{retrieval-into-budget} methods that expose a compatible fixed-budget interface ($\text{retrieve}(\text{query}, B) \to \text{context}$): Flat Vector, Maximal Marginal Relevance (MMR), Tree Abstraction (RAPTOR-style recursive clustering with a coarse-to-fine budget expansion), Disjoint Hierarchy, and Coverage baselines. Full agent-memory systems such as MemGPT \citep{packer2024memgptllmsoperatingsystems}, A-MEM \citep{xu2025amemagenticmemoryllm}, HiAgent \citep{hu-etal-2025-hiagent}, and Zep \citep{rasmussen2025zeptemporalknowledgegraph} operate at a different layer: their methodological content is the paging (MemGPT), note-linking (A-MEM), subgoal organization (HiAgent), or temporal-graph (Zep) logic that surrounds retrieval, not the budgeted retrieval step itself. A head-to-head comparison would therefore either strip them to a retrieval core---discarding the very machinery that defines them---or wrap their upstream logic around OSU-Mem's output, and neither isolates the retrieval mechanism this paper studies. We accordingly limit claims to the retrieval layer and leave a fair, fully controlled comparison against full memory systems to future work, which we discuss in Section~\ref{sec:discussion}.

\section{OSU-Mem: Method}
\label{sec:method}

\begin{figure}[h]
\centering
\includegraphics[width=0.6\linewidth]{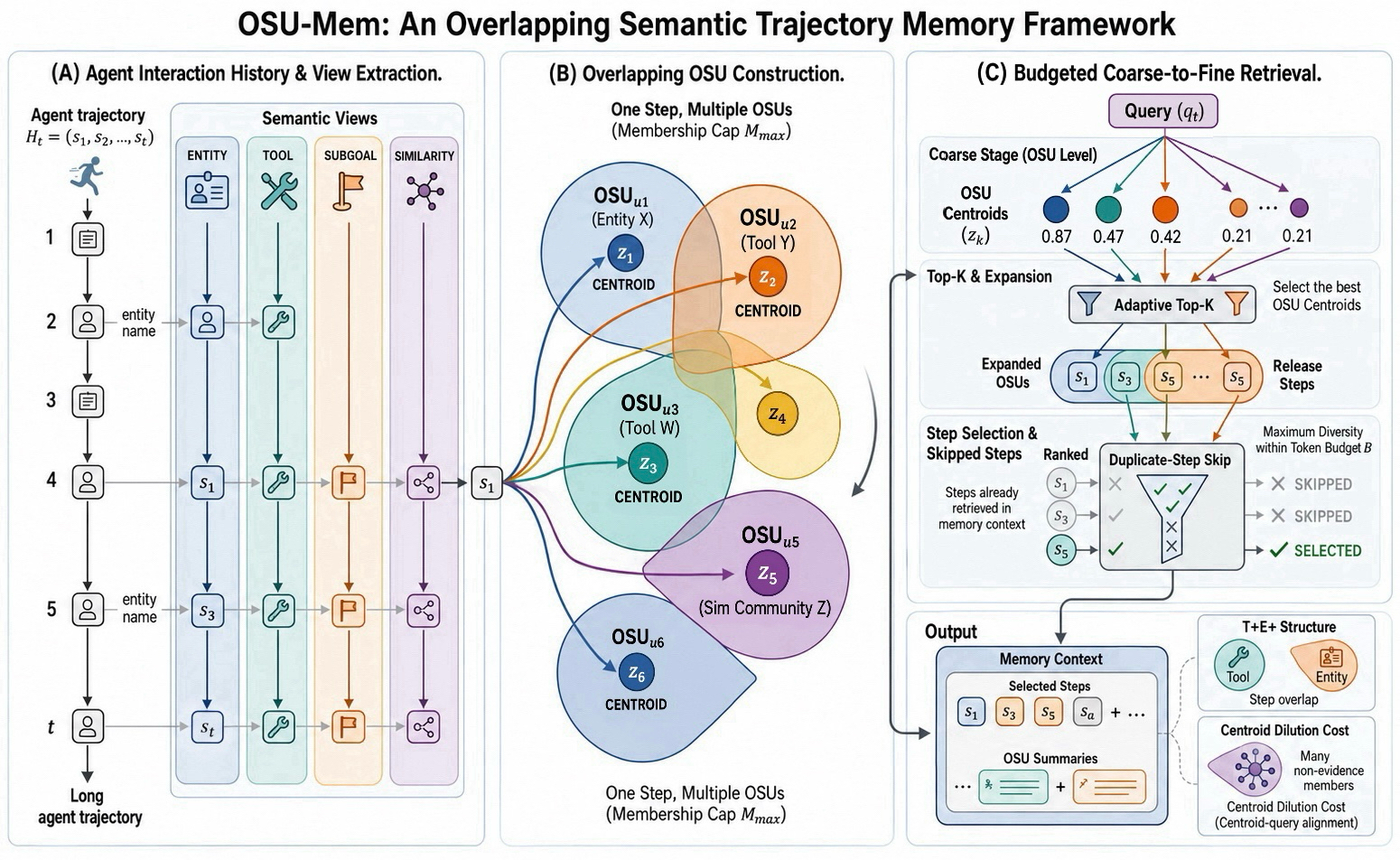}
\caption{Conceptual overview of OSU-Mem.
\textbf{(A)} The trajectory $\mathcal{H}_t = (s_1, \ldots, s_t)$ is processed through four
semantic views (entity, tool, subgoal, similarity).
\textbf{(B)} The views form an overlapping OSU pool where one step (e.g., $s_1$) may join
several OSUs (cap $M_{\max}$); each OSU stores a centroid $\mathbf{z}_k$.
\textbf{(C)} The query $q_t$ scores centroids and selects a top-$K$, which are expanded
into steps and packed into budget $B$ (with duplicate-step skipping) to form the memory
context. Insets recall the T+E+ and centroid-dilution mechanisms analyzed in
Section~\ref{sec:theory}.}
\label{fig:overview}
\end{figure}

Figure~\ref{fig:overview} gives a conceptual overview of the full pipeline; the remainder of this section formalizes each stage.

\subsection{Problem Formulation}
\label{sec:method-problem}

At decision time $t$, an agent has accumulated a trajectory $\mathcal{H}_t = (s_1, \ldots, s_t)$ of steps, each with an embedding $\mathbf{h}_i$ and a token cost $c_i$. A memory system receives a query embedding $q_t$ and a strict token budget $B$, and must return a set of retrieved items $\mathcal{M}(q_t) \subseteq \mathcal{H}_t$ with $\sum_{i \in \mathcal{M}} c_i \le B$.

\subsection{Step Representation}
\label{sec:method-step}

In the synthetic benchmark (Section~\ref{sec:datasets}), no text canonicalization or real tokenization is performed. Each step is assigned a 64-dimensional embedding from weighted mixtures of deterministic basis vectors: event-related steps use event ($.55$), entity ($.20$), tool ($.15$), and subgoal ($.10$) factors with noise $\sigma{=}.55$, while background steps omit the event factor and use entity ($.35$), tool ($.35$), and subgoal ($.30$) with $\sigma{=}.45$. Query embeddings emphasize the event basis with weights $.70$/$.15$/$.10$/$.05$ and lower noise $\sigma{=}.25$. Token costs are drawn from $\mathcal{N}(15, 3^2)$ clipped at 5. This gives full control over the similarity landscape and isolates the retrieval mechanism from encoder confounds.

In both $\tau$-bench and ToolBench real-trajectory evaluations, steps are encoded with a frozen \texttt{all-MiniLM-L6-v2} Sentence-BERT (SBERT) encoder (384-d, $\ell_2$-normalized) and token costs are estimated as $\max(5, \min(40, \lfloor |\text{text}|/4\rfloor))$.

\subsection{Overlapping Semantic Unit Construction}
\label{sec:method-osu}

An OSU $u_k \subseteq \{1, \ldots, t\}$ is a subset of step indices grouped by a shared semantic view. Unlike disjoint segmentation, OSU-Mem allows $u_j \cap u_k \neq \emptyset$, so that the same step can belong to multiple OSUs. We construct OSUs from four views computed from step annotations:

\paragraph{Entity view.} For each entity string appearing in step annotations, an OSU is constructed containing all steps that mention that entity. Entity strings are extracted from tool-call argument dictionaries (string values of length 2 to 80 and numeric values). To prevent degenerate single-topic OSUs on long trajectories, each entity-keyed group is split whenever consecutive occurrences are separated by more than $\Delta_{\text{ent}} = \max(60, 8\Delta_{\text{gap}})$ steps, and each resulting segment is capped at $S_{\max} = 20$ steps. Entity keys with fewer than $f_{\min} = 2$ associated steps are pruned since they provide no retrieval advantage over Flat Vector. The large temporal-gap threshold reflects the observation that entities (account numbers, destination names) tend to persist across long segments of agent interaction.

\paragraph{Tool view.} For each non-boilerplate tool signature, an OSU groups all steps invoking that tool. Temporal-gap splitting uses a tighter threshold $\Delta_{\text{tool}} = \max(12, 2\Delta_{\text{gap}})$ because a long gap between two uses of the same tool is usually a semantic boundary. The same size cap $S_{\max}$ and minimum-size filter $f_{\min}$ apply.

\paragraph{Subgoal view.} For each subgoal segment (identified from explicit labels or user-turn boundaries), a subgoal OSU is formed. The splitting threshold is the most aggressive: $\Delta_{\text{sg}} = \max(8, \Delta_{\text{gap}})$, since subgoals are the shortest semantic unit.

\paragraph{Similarity view.} Community merging on the step-embedding $k$-nearest-neighbor ($k$-NN) graph forms similarity OSUs that do not require explicit annotations. This view serves as a fallback for steps that lack entity or tool annotations.

\subsection{Membership Management and View Priority}
\label{sec:method-membership}

Because the four views may independently assign a step to multiple OSUs, a global membership cap $M_{\max} = 5$ limits the number of OSUs any single step may belong to. The value $M_{\max} = 5$ corresponds to one membership slot per view plus one slack slot; ablation shows no further gain beyond this value (Appendix~\ref{app:synth-ablation}). When the cap is exceeded, memberships are pruned according to a fixed view-priority rule:
\[
\text{entity} > \text{tool} > \text{similarity} > \text{subgoal}.
\]
Entity is prioritized because entity annotations are the most semantically stable cross-step identifiers in agent trajectories (account numbers, application programming interface (API) endpoints, file names persist across long segments). The $\tau$-bench cell analysis (Section~\ref{sec:cross-tab}) subsequently confirms that entity-keyed OSUs carry the strongest T+E+ signal. The distinction between $M_{\max} = 1$ (disjoint) and $M_{\max} \geq 2$ (overlap) is central to our method: in disjoint mode, each step has exactly one retrieval handle, while overlap provides multiple independent retrieval pathways. This distinction is analyzed formally in Section~\ref{sec:theory}.

\subsection{OSU Centroid and Representation}
\label{sec:method-centroid}

Each OSU $u_k$ is assigned a centroid embedding $\mathbf{z}_k$ defined as the $\ell_2$-normalized mean of its member step embeddings:
\begin{equation}
\mathbf{z}_k = \frac{\bar{\mathbf{h}}_k}{\|\bar{\mathbf{h}}_k\|_2}, \qquad \bar{\mathbf{h}}_k = \frac{1}{|u_k|}\sum_{i \in u_k}\mathbf{h}_i.
\label{eq:centroid}
\end{equation}
The centroid serves as the coarse-stage representation: a query's similarity to $\mathbf{z}_k$ determines whether OSU $u_k$ is expanded during retrieval. OSUs from different views may exhibit containment relationships (e.g., an entity OSU nested within a broader subgoal OSU); we do not model them as explicit parent-child structures and rely on shared step membership to preserve accessibility.

\subsection{Budgeted Query-Adaptive Expansion}
\label{sec:method-retrieval}

Given query $q_t$ and budget $B$, OSU-Mem (Algorithm~\ref{alg:osu}) scores all OSUs by $\mathbf{z}_k^\top q_t$, selects the top-$K$ (adaptively sized as $K = \min(|\{u_k\}|,\; \max(8,\; |\{u_k\}|/10,\; B/30))$, i.e., roughly $10\%$ of the OSU pool with a floor of 8 and a budget-proportional term), and then performs \emph{cost-normalized greedy expansion}. Rather than processing OSUs in a fixed score-descending order, each iteration selects the remaining top-$K$ OSU that maximizes $(r_k + \lambda_{\mathrm{nov}} \cdot \nu_k) / (\hat{c}_k + \epsilon)$, where $r_k = \mathbf{z}_k^\top q_t$ is the relevance score, $\nu_k$ is a blended novelty score of the OSU's candidates against the already-selected set (Appendix~\ref{app:method}), $\hat{c}_k$ is the total token cost of the OSU's top-$C$ candidates (with $C = 8$), and $\lambda_{\mathrm{nov}} = .3$. The novelty term is \emph{budget-gated}: it activates only when $B \ge B_{\mathrm{nov}} = 192$, because at very small budgets the selection pool is too small for novelty scoring to provide useful signal. When the selected OSU is expanded, its member steps are added to the output set in descending step-query similarity until the budget is reached. A duplicate-step check ensures that overlap among OSUs does not waste budget on redundant retrievals: if step $i$ has already been selected via an earlier OSU, it is skipped at zero cost. The resulting context is mixed-granularity: coarse OSU summaries provide broad cues, while expanded member steps provide precise evidence. Selected steps are packed in descending step--query similarity order before OSU summaries; under this default policy, summaries account for $<1\%$ of packed tokens, so the retrieval benefit derives from improved step selection via OSU-guided routing rather than from summaries.

\begin{algorithm}[t]
\small
\caption{OSU-Mem retrieval under a strict token budget.}
\label{alg:osu}

\KwIn{trajectory steps $\{s_i\}$, query embedding $q_t$, token budget $B$, OSU pool $\{u_k\}$ with centroids $\{\mathbf{z}_k\}$}
\KwOut{$\mathcal{S}$}

Score each OSU: $\sigma_k \leftarrow \mathbf{z}_k^\top q_t$, sort descending\;
Select top-$K$ OSUs adaptively (fraction of $|\{u_k\}|$, clamped)\;
Initialize selected set $\mathcal{S} \leftarrow \emptyset$, consumed budget $b \leftarrow 0$, available $\leftarrow$ top-$K$ OSU set\;

\While{available $\neq \emptyset$ and $b < B$}{
  \ForEach{OSU $u_k$ in available}{
    Compute relevance $r_k \leftarrow \mathbf{z}_k^\top q_t$\;
    Compute candidate cost $\hat{c}_k \leftarrow \sum_{i \in \text{top-}C} c_i$\tcp*[r]{$C{=}8$ candidates per OSU}
    
    \eIf{$B \ge B_{\mathrm{nov}}$}{
      Compute blended novelty $\nu_k$ of $u_k$'s candidates against $\mathcal{S}$\;
    }{
      $\nu_k \leftarrow 0$\;
    }
    
    $\text{score}_k \leftarrow (r_k + \lambda_{\mathrm{nov}} \cdot \nu_k) / (\hat{c}_k + \epsilon)$\;
  }

  $k^* \leftarrow \arg\max_k \text{score}_k$; remove $k^*$ from available\;
  Rank members $i \in u_{k^*}$ by $\mathbf{h}_i^\top q_t$\;

  \ForEach{candidate step $i$}{
    \If{$i \notin \mathcal{S}$ and $b + c_i \le B$}{
      $\mathcal{S} \leftarrow \mathcal{S} \cup \{i\}$, $b \leftarrow b + c_i$\;
    }
  }
}

\Return{$\mathcal{S}$}\tcp*[r]{OSU summaries appended with remaining budget}

\end{algorithm}

\subsection{Hyperparameters}
\label{sec:method-hyperparams}

All OSU-Mem hyperparameters, their default values, and the rationale for each choice are listed in Appendix~\ref{app:method} (Table~\ref{tab:hyperparams}). All values are fixed across the three benchmarks (synthetic, $\tau$-bench, ToolBench); none was tuned per-benchmark. The single per-benchmark deviation is the retrieval algorithm: the synthetic and $\tau$-bench benchmarks use the default coarse-to-fine algorithm (Algorithm~\ref{alg:osu}), while the ToolBench benchmark uses a coverage-guided greedy variant because SBERT-encoded ToolBench has entangled OSU centroids on which coarse routing underperforms (Section~\ref{sec:overall}).

None of these defaults was tuned on the $\tau$-bench T+E+ cell. We fixed them on the synthetic benchmark before any real-trajectory analysis, to prevent the cell-conditional claim from being an artifact of hyperparameter overfitting to $\tau$-bench's T+E+ sub-population.

\section{Theoretical Analysis}
\label{sec:theory}

In this section, we formalize three properties of overlapping OSU construction that underpin the empirical findings in Section~\ref{sec:results}. We first show that overlap increases retrieval reachability relative to disjoint assignment---strictly so under an independence model, and non-strictly without it (Proposition~\ref{prop:reachability})---then characterize the amplified benefit under the T+E+ evidence-structure condition (Proposition~\ref{prop:tpep}), and finally analyze the centroid-dilution cost that explains the T$-$E$-$ disadvantage at large budgets (Proposition~\ref{prop:dilution}).

\subsection{Retrieval Reachability under Overlap}
\label{sec:theory-reach}

\begin{definition}[Retrieval reachability]
\label{def:reachability}
Given a query $q$, a step $s_i$ is \emph{retrievable} if there exists an OSU $u_k$ such that $i \in u_k$ and $u_k$ is selected at the coarse stage (i.e., $\mathbf{z}_k^\top q \ge \tau$ for some selection threshold $\tau$, or equivalently $u_k$ ranks in the top-$K$). We denote the retrieval reachability of $s_i$ as $P_{\mathrm{reach}}(s_i \mid q) = P(\exists\, u_k \ni i : u_k \text{ selected})$.
\end{definition}

\begin{proposition}[Overlap increases reachability; stylized selection model]
\label{prop:reachability}
Let $\mathcal{V} = \{v_1, \ldots, v_V\}$ denote the set of semantic views. For a step $s_i$ with non-trivial annotations in views $\mathcal{V}_i \subseteq \mathcal{V}$, let $p_{v}(q) = P(u_k^{(v)} \text{ selected} \mid i \in u_k^{(v)})$ denote the probability that the OSU containing $s_i$ under view $v$ is selected at the coarse stage, where we model view-level selections as independent events. Under disjoint assignment ($M_{\max} = 1$), step $s_i$ is assigned to exactly one view $v^* \in \mathcal{V}_i$ by the priority rule, yielding
\begin{equation}
P_{\mathrm{reach}}^{\mathrm{disj}}(s_i \mid q) = p_{v^*}(q).
\label{eq:reach-disj}
\end{equation}
Under overlap ($M_{\max} \ge |\mathcal{V}_i|$), step $s_i$ belongs to one OSU per applicable view, yielding
\begin{equation}
P_{\mathrm{reach}}^{\mathrm{over}}(s_i \mid q) = 1 - \prod_{v \in \mathcal{V}_i} \bigl(1 - p_v(q)\bigr).
\label{eq:reach-over}
\end{equation}
When $|\mathcal{V}_i| \ge 2$, at least two views have $p_v(q) > 0$, and $p_{v^*}(q) < 1$, we have $P_{\mathrm{reach}}^{\mathrm{over}}(s_i \mid q) > P_{\mathrm{reach}}^{\mathrm{disj}}(s_i \mid q)$. (The condition $p_{v^*} < 1$ is non-restrictive in practice since coarse top-$K$ selection with $K \ll |\{u_k\}|$ never guarantees selection of any individual OSU.) The \emph{strict} inequality relies on the independence model of Equation~\ref{eq:reach-over}; without independence, a weaker union bound still yields the non-strict $P_{\mathrm{reach}}^{\mathrm{over}} \ge P_{\mathrm{reach}}^{\mathrm{disj}}$ (see the note following the proof).
\end{proposition}

\begin{proof}
Under disjoint assignment, $s_i$ participates in exactly one OSU $u_k^{(v^*)}$, so it is reachable only if that single OSU is selected, giving $P_{\mathrm{reach}}^{\mathrm{disj}} = p_{v^*}$. Under overlap, $s_i$ participates in $|\mathcal{V}_i|$ OSUs (one per view), and reachability requires \emph{at least one} to be selected. By the independence assumption, $P_{\mathrm{reach}}^{\mathrm{over}} = 1 - \prod_{v \in \mathcal{V}_i}(1 - p_v)$. Now, $P_{\mathrm{reach}}^{\mathrm{over}} - P_{\mathrm{reach}}^{\mathrm{disj}} = 1 - \prod_{v}(1 - p_v) - p_{v^*}$. Note that $1 - \prod_{v}(1 - p_v) \ge 1 - (1 - p_{v^*})(1 - p_{v'}) = p_{v^*} + p_{v'} - p_{v^*}p_{v'}$ for any $v' \neq v^*$ with $p_{v'} > 0$. Since $p_{v'} - p_{v^*}p_{v'} = p_{v'}(1 - p_{v^*}) > 0$, the strict inequality follows.
\end{proof}

\paragraph{Connection to experiments.} On ToolBench, under disjoint construction ($M_{\max}{=}1$), $13.2\%$ of evidence steps belong to \emph{no} OSU because view-priority pruning forces each step into exactly one view, and the minimum-size filter then deletes views with fewer than two members. These evicted evidence steps have $P_{\mathrm{reach}}^{\mathrm{disj}} = 0$. Under overlap ($M_{\max}{=}5$), $0\%$ of evidence is evicted---directly instantiating the strict inequality in Proposition~\ref{prop:reachability}.

\paragraph{Note on the independence assumption.} The independence model in Equation~\ref{eq:reach-over} is an approximation. In practice, all views' OSUs compete for the same top-$K$ pool, creating negative correlation: when one view's OSU ranks highly, it displaces others (including sibling-view OSUs of the same step). Additionally, switching from $M_{\max}{=}1$ to $M_{\max}{=}5$ increases the total OSU count, which may change each OSU's absolute selection probability even though $K$ adapts proportionally. Both effects cause the independence formula to \emph{overestimate} the magnitude of the reachability gain. The \emph{direction} of the gain, however, is robust: a weaker, assumption-free union bound gives $P_{\mathrm{reach}}^{\mathrm{over}} \ge \max_{v} p_v \ge p_{v^*} = P_{\mathrm{reach}}^{\mathrm{disj}}$, which preserves the non-strict inequality without any independence assumption. The strict inequality additionally requires that at least one non-priority view has $p_v > 0$ and $p_{v^*} < 1$; under the independence model these conditions are sufficient, and even without independence they suffice unless the view-selection events are pathologically coupled (e.g., perfectly negatively correlated such that $u_{v'}$ is never selected when $u_{v^*}$ is not). In the competitive top-$K$ setting, negative correlation is the norm but perfect coupling is not, so the strict inequality holds in practice even though its formal guarantee relies on the independence model.

\subsection{Amplified Benefit under T+E+ Evidence Structure}
\label{sec:theory-tpep}

We now formalize why the overlap advantage is amplified when evidence steps share tool signatures and entities (the T+E+ condition).

\begin{definition}[T+E+ evidence structure]
\label{def:tpep}
Given a query $q$ with ground-truth evidence set $\mathcal{E}_q = \{s_{e_1}, \ldots, s_{e_n}\}$, we say $\mathcal{E}_q$ satisfies the \emph{T+E+} condition if (i) there exist distinct $s_{e_i}, s_{e_j} \in \mathcal{E}_q$ sharing at least one tool signature (the T+ condition), and (ii) there exist distinct $s_{e_k}, s_{e_l} \in \mathcal{E}_q$ sharing at least one entity string (the E+ condition). The pairs $(s_{e_i}, s_{e_j})$ and $(s_{e_k}, s_{e_l})$ need not be the same.
\end{definition}

The remaining three cells of the $2{\times}2$ cross-tabulation are defined analogously: \emph{T+E$-$} denotes evidence satisfying T+ but not E+ (tool sharing without entity sharing), \emph{T$-$E+} denotes evidence satisfying E+ but not T+ (entity sharing without tool sharing), and \emph{T$-$E$-$} denotes evidence satisfying neither condition (fully heterogeneous evidence with no pairwise tool or entity sharing).

\begin{proposition}[T+E+ amplifies overlap recall]
\label{prop:tpep}
Let $\mathcal{E}_q$ satisfy the T+E+ condition with $n = |\mathcal{E}_q|$ evidence steps. Define the \emph{cluster recall} of an OSU as $\mathrm{CR}(u_k) = |u_k \cap \mathcal{E}_q| / n$. Under the T+E+ condition, and assuming the evidence steps with shared annotations fall within the temporal-gap splitting thresholds ($\Delta_{\mathrm{ent}}$, $\Delta_{\mathrm{tool}}$) so that sharing implies co-membership in the same view-keyed OSU, there exist at least one entity-keyed OSU $u^{(\mathrm{ent})}$ with $|u^{(\mathrm{ent})} \cap \mathcal{E}_q| \ge 2$ (from the shared entity) and at least one tool-keyed OSU $u^{(\mathrm{tool})}$ with $|u^{(\mathrm{tool})} \cap \mathcal{E}_q| \ge 2$ (from the shared tool). These two OSUs may bind \emph{different subsets} of evidence. Let $\mathcal{E}^{(\mathrm{ent})} = u^{(\mathrm{ent})} \cap \mathcal{E}_q$ and $\mathcal{E}^{(\mathrm{tool})} = u^{(\mathrm{tool})} \cap \mathcal{E}_q$. Fix the entity OSU $u^{(\mathrm{ent})}$ and consider the evidence reachable through it under both regimes:

\emph{Under overlap} ($M_{\max} \ge 2$), the evidence steps in $u^{(\mathrm{ent})}$ also retain their tool-view memberships, so $u^{(\mathrm{tool})}$ coexists in the index. A single coarse-stage selection of $u^{(\mathrm{ent})}$ reaches $\mathcal{E}^{(\mathrm{ent})}$; additionally, $u^{(\mathrm{tool})}$ provides a second coarse-stage entry point reaching $\mathcal{E}^{(\mathrm{tool})}$. The evidence reachable through at least one of these two pathways is
\begin{equation}
\mathrm{CR}_{\mathrm{combined}}^{\mathrm{over}} = \frac{|\mathcal{E}^{(\mathrm{ent})} \cup \mathcal{E}^{(\mathrm{tool})}|}{n} \ge \frac{2}{n}.
\label{eq:cluster-recall-overlap}
\end{equation}

\emph{Under disjoint} ($M_{\max} = 1$) with entity~$>$~tool priority, assuming all evidence steps carry at least one entity annotation (which holds in practice since entities are extracted from tool-call argument dictionaries), every evidence step is assigned to its entity-keyed OSU, provided that its entity group contains at least $f_{\min} = 2$ members in the trajectory (singleton entity groups are pruned; however, T+E+ guarantees at least one entity group with $\ge 2$ evidence members, so $u^{(\mathrm{ent})}$ survives). Consequently, tool-keyed OSUs lose their evidence members. Selecting the same $u^{(\mathrm{ent})}$ at the coarse stage reaches only $\mathcal{E}^{(\mathrm{ent})}$:
\begin{equation}
\mathrm{CR}^{\mathrm{disj}}(u^{(\mathrm{ent})}) = \frac{|\mathcal{E}^{(\mathrm{ent})}|}{n}.
\label{eq:cluster-recall-disj}
\end{equation}
When $\mathcal{E}^{(\mathrm{tool})} \not\subseteq \mathcal{E}^{(\mathrm{ent})}$, i.e., the tool-view OSU binds evidence steps with entities distinct from those in $u^{(\mathrm{ent})}$, we have $\mathrm{CR}_{\mathrm{combined}}^{\mathrm{over}} > \mathrm{CR}^{\mathrm{disj}}(u^{(\mathrm{ent})})$.
\end{proposition}

\paragraph{Scope of comparison.} Proposition~\ref{prop:tpep} compares evidence reachable through a \emph{fixed entity OSU} $u^{(\mathrm{ent})}$ under both regimes, plus the additional tool-view pathway available only under overlap. This is a per-coarse-selection comparison. In full top-$K$ retrieval, disjoint mode can also select multiple entity-keyed OSUs whose union may cover more evidence; however, it cannot access tool-view OSUs at all, since entity~$>$~tool priority strips all evidence from them. The qualitative insight, that overlap provides a cross-view retrieval pathway unavailable under disjoint, holds regardless of $K$.

\begin{proof}
Under T+E+, at least one pair of evidence steps shares an entity and at least one pair shares a tool (Definition~\ref{def:tpep}). Entity sharing guarantees $|u^{(\mathrm{ent})} \cap \mathcal{E}_q| \ge 2$, and tool sharing guarantees $|u^{(\mathrm{tool})} \cap \mathcal{E}_q| \ge 2$, giving the lower bound in Equation~\ref{eq:cluster-recall-overlap}. Under disjoint with entity~$>$~tool priority, each evidence step possessing an entity annotation is assigned to an entity-keyed OSU and removed from all tool-keyed OSUs. Selecting $u^{(\mathrm{ent})}$ therefore reaches exactly $\mathcal{E}^{(\mathrm{ent})}$. Under overlap, the same entity OSU still reaches $\mathcal{E}^{(\mathrm{ent})}$, but evidence steps in $u^{(\mathrm{ent})}$ also belong to $u^{(\mathrm{tool})}$, which additionally reaches $\mathcal{E}^{(\mathrm{tool})}$. When $\mathcal{E}^{(\mathrm{tool})}$ includes steps with entities distinct from those in $u^{(\mathrm{ent})}$, these steps are scattered across \emph{other} entity-keyed OSUs under disjoint but are united under the tool-view OSU under overlap, making the union strictly larger.
\end{proof}

\paragraph{Connection to experiments.} The ToolBench evidence-OSU statistics (Section~\ref{sec:overall}) confirm this quantitatively: on ToolBench, under overlap each evidence step belongs to $2.16$ OSUs on average (median $2$), whereas under disjoint the average drops to $.87$ (median $1$), a $2.49\times$ increase. The additional tool-view and entity-view OSU memberships under overlap provide independent retrieval handles that bind different subsets of the evidence cluster, enabling coarse selection from multiple angles.

\subsection{Centroid Dilution and the T$-$E$-$ Disadvantage}
\label{sec:theory-dilution}

The overlap construction is not uniformly beneficial. Under the T$-$E$-$ condition, where evidence steps lack shared annotations with one another, evidence does not cluster within any single OSU. The coarse-to-fine routing then adds an indirection layer (centroid matching) that is less precise than Flat Vector's direct step-level cosine ranking.

\begin{proposition}[Centroid dilution under T$-$E$-$]
\label{prop:dilution}
Consider an OSU $u_k$ with $|u_k| = m$ member steps, of which $m_e$ are evidence steps and $m - m_e$ are non-evidence steps. Let $\bar{\mathbf{h}}_e$ and $\bar{\mathbf{h}}_{\neg e}$ denote the mean embeddings of the evidence and non-evidence members, respectively. The centroid-query alignment is
\begin{equation}
\mathbf{z}_k^\top q = \frac{1}{m\|\bar{\mathbf{h}}_k\|}\left(m_e\, \bar{\mathbf{h}}_e^\top q + (m - m_e)\, \bar{\mathbf{h}}_{\neg e}^\top q\right).
\label{eq:dilution}
\end{equation}
When evidence steps are semantically distinct from non-evidence members ($\bar{\mathbf{h}}_e^\top q \gg \bar{\mathbf{h}}_{\neg e}^\top q$), the centroid-query alignment tends to decrease as the ratio $m_e / m$ decreases. Under the T$-$E$-$ condition, no shared annotation binds evidence steps into the same OSU, so $m_e / m$ tends to be small for any single OSU. This reduces the coarse-stage score $\sigma_k = \mathbf{z}_k^\top q$, making the OSU less likely to be selected in the top-$K$.
\end{proposition}

\begin{proof}
Direct from Equation~\ref{eq:centroid}: $\mathbf{z}_k^\top q = \bar{\mathbf{h}}_k^\top q / \|\bar{\mathbf{h}}_k\|$, where $\bar{\mathbf{h}}_k = (m_e \bar{\mathbf{h}}_e + (m-m_e)\bar{\mathbf{h}}_{\neg e})/m$. Expanding gives Equation~\ref{eq:dilution}. The numerator $m_e\, \bar{\mathbf{h}}_e^\top q + (m - m_e)\, \bar{\mathbf{h}}_{\neg e}^\top q$ is $m$ times a convex combination of $\bar{\mathbf{h}}_e^\top q$ and $\bar{\mathbf{h}}_{\neg e}^\top q$ with weights $m_e/m$ and $(m-m_e)/m$. As the evidence fraction $m_e/m$ decreases, this convex combination shifts toward $\bar{\mathbf{h}}_{\neg e}^\top q$. The denominator $\|\bar{\mathbf{h}}_k\|$ also varies with $m_e/m$ and is not generally monotone (it depends on the angle between $\bar{\mathbf{h}}_e$ and $\bar{\mathbf{h}}_{\neg e}$), so the overall effect on $\mathbf{z}_k^\top q$ need not be strictly monotone in $m_e/m$ for intermediate values. However, in the limit $m_e/m \to 0$, the centroid aligns with the non-evidence mean ($\mathbf{z}_k \to \bar{\mathbf{h}}_{\neg e} / \|\bar{\mathbf{h}}_{\neg e}\|$), and the coarse-stage score falls below what Flat Vector achieves by directly computing $\mathbf{h}_i^\top q$ for each evidence step individually. The practical takeaway is that OSUs with a small evidence fraction tend to receive low coarse-stage scores, though the precise dilution trajectory depends on the geometric relationship between evidence and non-evidence embeddings.
\end{proof}

\paragraph{Budget interaction.} The T$-$E$-$ disadvantage grows with budget $B$ because at large budgets, Flat Vector has enough retrieval slots to cover most evidence through direct cosine ranking (all evidence with $\mathbf{h}_i^\top q$ above a natural threshold is retrieved), while OSU-Mem's diluted centroids may misdirect the coarse-stage budget toward low-evidence OSUs. At small budgets, the forced selection among very few slots makes any structured index (including diluted OSUs) potentially beneficial simply by providing \emph{any} non-trivial ordering. This explains the empirical pattern in Table~\ref{tab:cell-crosstab-multi}: the T$-$E$-$ gap is $+.8$ percentage points (pp) at $B{=}64$ (negligible), $-2.3$ pp at $B{=}128$, $-12.9$ pp at $B{=}256$, and $-36.4$ pp at $B{=}512$ (strongly negative).

\subsection{Summary}
The three propositions together predict a bimodal performance pattern that is consistent with what the experiments observe: overlap helps on T+E+ evidence (Propositions~\ref{prop:reachability} and~\ref{prop:tpep}) and hurts on T$-$E$-$ evidence at moderate-to-large budgets (Proposition~\ref{prop:dilution}), with the crossover budget depending on the evidence-to-noise ratio in the OSU pool.

\section{Experimental Settings}
\label{sec:settings}

\subsection{Datasets}
\label{sec:datasets}

We evaluate OSU-Mem on three benchmarks spanning a controlled synthetic environment and two real-trajectory sources. Table~\ref{tab:dataset-stats} summarizes their key statistics.

\paragraph{Synthetic benchmark.} We generate 200 episodes (197 with valid evaluation queries, 907 total queries) with 150 to 296 steps each. Step embeddings are 64-dimensional basis-vector constructions. Event-related steps use weights event ($.55$), entity ($.20$), tool ($.15$), subgoal ($.10$) with Gaussian noise $\sigma{=}.55$; non-event background steps omit the event factor and use entity ($.35$), tool ($.35$), subgoal ($.30$) with $\sigma{=}.45$. Query embeddings emphasize the event factor with weights event ($.70$), entity ($.15$), tool ($.10$), subgoal ($.05$) and lower noise $\sigma{=}.25$. Token costs are drawn from $\mathcal{N}(15, 3^2)$ clipped at 5. This benchmark provides full control over the similarity landscape and the evidence structure. Because all evidence steps for a query are drawn from a single shared semantic event, they carry the same entity and subgoal annotations and the same dominant event embedding factor: the E+ condition (shared entity) and shared-subgoal condition hold for $100\%$ of queries by construction. Tool sharing is weaker by design---about $20\%$ of event steps draw a random tool and each tool signature appends a random argument-template suffix---so evidence steps pairwise share an identical tool \emph{signature} for $39.3\%$ of queries (and share the same tool \emph{action}, ignoring the argument template, for $83.4\%$). The synthetic benchmark is therefore an \emph{E+-by-construction} regime with partial T+ sharing rather than a uniformly T+E+ one; it isolates a shared-structure retrieval regime in which the entity/subgoal annotations and dominant event signal are guaranteed by construction, while tool sharing is present only probabilistically. We retain it as the controlled testbed in which the cross-view binding mechanism is isolated from encoder confounds.

\paragraph{$\tau$-bench.} Raw $\tau$-bench~\citep{yao2025taubench} trajectories have median length 29 steps, too short to exercise budgeted retrieval at meaningful scale, since $B{=}256$ trivially admits the entire trajectory. Concatenation is therefore a necessity imposed by the trajectory length available in $\tau$-bench rather than a design choice. We concatenate trajectories into groups with total length in $[150, 300]$ steps (125 concatenated episodes, 328 evaluation queries, concatenated trajectory length 256 to 311). Each concatenated episode has one \emph{focal source} whose instruction becomes the query, and focal evidence steps are the focal source's ground-truth actions (re-offset to the concatenated coordinate system), all temporally distant from the query time. Entity strings and tool signatures are \emph{namespace-prefixed per source episode} to prevent cross-source string collisions, a deployment pitfall we document in Appendix~\ref{app:taubench-pipeline}. An initial implementation without namespace prefixing \emph{reversed} the overlap-vs.-disjoint ordering because generic strings such as \texttt{"user"} and \texttt{"search"} bridged unrelated tasks. After prefixing, the cross-source sharing rate drops to the level observed within genuine same-source pairs, and the cross-tab analysis in Section~\ref{sec:cross-tab} is a clean within-source measurement. Focal query embeddings are replaced with an $\ell_2$-normalized noisy mean ($\sigma{=}.25$) of focal evidence embeddings, because the original $\tau$-bench queries are user-turn instructions already consumed earlier in the trajectory and lack a natural standalone embedding. This realignment is applied identically to all methods and does not give OSU-Mem preferential access; however, it may inflate absolute recall for any cosine-retrieval method. Our primary mitigation is the ToolBench experiment (Section~\ref{sec:overall}), where queries are the \emph{original natural-language queries} embedded verbatim, and the overlap-vs.-disjoint gap reproduces with the same sign and comparable magnitude. This concatenation pipeline introduces a known limitation (Section~\ref{sec:discussion}, Limitation~v). Because both the concatenation and the noisy-mean query construction depart from natural usage, we treat the $\tau$-bench results as a constructed retrieval probe, \emph{not} an end-to-end natural-query benchmark; the natural-language-query evidence comes from ToolBench, whose queries are embedded verbatim (Section~\ref{sec:overall}).

\paragraph{ToolBench.} We use ToolBench~\citep{qin2024toolllmfacilitatinglargelanguage} to build a \emph{controlled mechanism probe} rather than a natural agent-memory benchmark: the construction deliberately restricts the evidence distribution to the favorable T+E+ cell so that the overlap mechanism can be measured in isolation. Individual ToolBench trajectories are short (3 to 5 tool-calling steps) and cannot exercise budgeted retrieval directly. For each of 200 episodes we (1) select a target trajectory from the G1/G2 splits (single-tool and multi-tool categories, respectively, as defined in the original ToolBench dataset) and designate 3 of its steps as evidence, (2) draw a ${\approx}500$-step background of $75\%$ confounders (steps from top-30 query-embedding nearest-neighbor trajectories) plus $25\%$ unrelated steps, and (3) place the 3 evidence steps at three evenly-spaced positions. We designate evidence size 3 as the primary setting because it is the smallest nontrivial planted-evidence cluster that supports pairwise tool/entity sharing beyond a single pair. We report sensitivity over evidence sizes $\{2,3,4\}$ in Appendix~\ref{app:toolbench-sensitivity}, where the overlap advantage remains directionally consistent. Results are averaged over 3 seeds (42, 123, 456). This planted-retrieval construction is \emph{structurally T+E+ by design}: the 3 planted evidence steps always come from the \emph{same} target trajectory; ToolBench function names follow the convention \texttt{<fname>\_for\_<api>}, so all evidence steps pairwise share a single tool signature (T+); and steps from the same trajectory call the same API with related argument dictionaries, producing high entity sharing (E+). Because the favorable cell is imposed by construction, this experiment measures the overlap mechanism under a known-favorable structure; it is not evidence of deployment utility on natural ToolBench agents, whose evidence distribution we do not control. Full construction details appear in Appendix~\ref{app:toolbench-construction}.

\begin{table}[t]
\centering
\small
\caption{Summary statistics of the three evaluation benchmarks.}
\label{tab:dataset-stats}
\begin{tabular}{lccc}
\toprule
 & \textbf{Synthetic} & \textbf{$\tau$-bench} & \textbf{ToolBench} \\
\midrule
Episodes & 197 & 125 & 200 $\times$ 3 seeds \\
Queries & 907 & 328 & 200 $\times$ 3 seeds \\
Steps per episode & 150--296 & 256--311 & $\approx$500 \\
Embedding dim & 64 & 384 (SBERT) & 384 (SBERT) \\
Evidence per query & 2.5 (mean) & 3.50 (mean) & 3 (planted) \\
T+E+ fraction & 39\% (E+ 100\% by constr.) & 17.1\% & $\approx$100\% (by construction) \\
Budgets tested & \{64, 128, 256, 512\} & \{64, 128, 256, 512\} & \{200, 400, 800, 1600\} \\
\bottomrule
\end{tabular}
\end{table}

\subsection{Implementation Details}
\label{sec:impl}

For the synthetic benchmark, all embeddings, OSU construction, and retrieval are computed in the 64-d basis-vector space. For both trajectory-derived benchmarks, steps are encoded with a frozen \texttt{all-MiniLM-L6-v2} SBERT encoder (384-d, $\ell_2$-normalized). OSU construction uses the hyperparameters in Table~\ref{tab:hyperparams}, fixed across all three benchmarks. The default retrieval algorithm is coarse-to-fine (Algorithm~\ref{alg:osu}) for synthetic and $\tau$-bench. For ToolBench, pilot measurements show that sibling-view OSU centroids have mutual cosine $> .6$ on average and coarse-routing recall falls below $.17$, consistent with all-MiniLM-L6-v2 entangling entity-level, tool-level, and semantic variation into a single dense vector. We therefore evaluate ToolBench under a \emph{coverage-guided greedy} retrieval variant: at each step the selected OSU is the one whose member step maximizes $s_i + \beta \cdot [\text{new OSU}]$, where $s_i$ is step-query cosine similarity and $\beta$ is a coverage bonus rewarding steps from previously-unused OSUs. Both overlap and disjoint indexes are scored by the identical retrieval procedure, and only the OSU-pool construction differs.

\subsection{Baselines and Evaluation Metrics}
\label{sec:baselines}

\paragraph{Baselines.} We compare OSU-Mem against both simple and structured retrieval-into-budget methods. \textbf{No Long-Term (LT) Memory} returns an empty context. \textbf{Budget-Truncated Replay} returns the most recent $B$ tokens of steps. \textbf{Sliding Window} uses a $B$-sized rolling recent window. \textbf{Periodic Summary} combines a recent window with periodic LLM-style summary of older chunks. \textbf{Flat Vector} returns the top-cosine steps by $\mathbf{h}_i^\top q$. \textbf{Flat $+$ MMR} applies $\lambda{=}.7$ diversity-penalized top-cosine. \textbf{Coverage-Aware} adds a $\beta$ bonus for covering new subgoals at the step level. \textbf{Disjoint Hierarchy} uses disjoint subgoal partitions with top-cosine within selected partitions. \textbf{Tree Abstraction} performs RAPTOR-style recursive summarization with a coarse-to-fine budget expansion. All baselines use the same token-cost estimation and packing policy.

To ensure fair comparison, all methods operate under the same token budget $B$ and use the same step embeddings. The ToolBench experiment additionally compares \textbf{OSU-Mem (overlap, $M{=}5$)} against \textbf{OSU-Mem (disjoint, $M{=}1$)} under coverage-guided retrieval to isolate the construction principle.

\paragraph{Evaluation metrics.} We report two primary retrieval-quality metrics. \textbf{Hit@2} is the fraction of queries for which at least two evidence steps are retrieved. We adopt Hit@2 as the primary criterion because our queries have multi-step ground-truth evidence (mean $3.5$ on $\tau$-bench, $3$ on ToolBench): Hit@1 is satisfied by retrieving a single evidence step and so does not measure whether enough of the cross-bound evidence is recovered to support a decision, whereas Hit@2 directly probes the multi-handle retrieval that the overlap mechanism is meant to provide. We report Recall alongside it as the finer-grained continuous measure; downstream answer accuracy is out of scope under our retrieval-layer focus (Section~\ref{sec:discussion}, Limitation~iv). \textbf{Recall} is the fraction of evidence steps retrieved, averaged across queries. Paired sign-flip permutation tests ($n_{\text{perm}}{=}10{,}000$) are used for all significance claims. For ToolBench, per-seed permutation $p$-values are Fisher-combined across seeds. We designate $B{=}256$ as the primary budget for $\tau$-bench and $B{=}400$ for ToolBench.

\section{Results and Analysis}
\label{sec:results}

\paragraph{Two contrasts: construction principle vs.\ full algorithm.} We separate two things OSU-Mem bundles. The \emph{construction principle} is whether a step may join multiple OSUs ($M_{\max}\ge2$, overlap) or exactly one ($M_{\max}{=}1$, disjoint); the \emph{retrieval algorithm} is the procedure that consumes the resulting index (coarse-to-fine, Algorithm~\ref{alg:osu}, or the coverage-guided greedy variant). The OSU-Mem-vs.-Flat~Vector comparison evaluates the \emph{full system} and cannot attribute a gain to either half; the overlap-vs.-disjoint comparison fixes the algorithm and varies only construction, isolating the construction principle. Our central mechanism claim---that cross-view binding determines when overlap helps---is about the construction principle, and the overlap-vs.-disjoint contrasts (the $\tau$-bench cross-tab's last column and the ToolBench $\beta$ sweep) are the comparisons that test it.

\subsection{Overall Results}
\label{sec:overall}

\paragraph{Synthetic benchmark.} At $B{=}256$, OSU-Mem achieves Recall $=.221$ and Hit@2 $=.084$ vs.\ $.158$ and $.052$ for the strongest baseline Flat $+$ MMR (Table~\ref{tab:synth-main}), which corresponds to relative improvements of $+39.9\%$ Recall and $+61.5\%$ Hit@2. Paired sign-flip permutation tests (10,000 permutations, $n{=}197$) confirm $p \le .002$ (two-sided) for both metrics against every retrieval baseline, and full significance details appear in Appendix~\ref{app:synth-sig}. OSU-Mem runs at 8.3\,ms amortized latency with 209\,KB storage, smaller and faster than Tree Abstraction (57.4\,ms, 363\,KB) at the 64-d regime (Appendix~\ref{app:synth-eff}); under 384-d SBERT used for both real-trajectory benchmarks, OSU-Mem has the highest retrieve-only latency at ${\approx}2.6$\,ms and the largest storage footprint, though its amortized cost remains below Tree Abstraction's (Appendix~\ref{app:taubench-eff}, Limitation~iii). Figure~\ref{fig:near-far} further breaks down recall by temporal distance, showing that OSU-Mem's advantage holds for both near and far evidence, confirming that the gain stems from non-local semantic structure rather than recency bias.

\begin{table}[t]
\centering
\small
\caption{Synthetic benchmark: evidence retrieval at $B{=}128$ and $B{=}256$ (primary budgets). OSU-Mem achieves the highest Hit@2 and Recall at both budgets. 197 episodes, 907 queries. Full 4-budget version in Appendix~\ref{app:synth-ablation}.}
\label{tab:synth-main}
\begin{tabular}{lcccc}
\toprule
\textbf{Method} & \multicolumn{2}{c}{$B{=}128$} & \multicolumn{2}{c}{$B{=}256$} \\
\cmidrule(lr){2-3}\cmidrule(lr){4-5}
 & Hit@2 & Recall & Hit@2 & Recall \\
\midrule
Budget-Trunc Replay & .000 & .004 & .000 & .149 \\
Sliding Window      & .000 & .000 & .000 & .136 \\
Periodic Summary    & .000 & .041 & .013 & .093 \\
Flat Vector         & .016 & .088 & .045 & .156 \\
Flat $+$ MMR        & .022 & .085 & .052 & .158 \\
Coverage-Aware      & .014 & .086 & .049 & .157 \\
Disjoint Hierarchy  & .014 & .084 & .044 & .150 \\
Tree Abstraction    & .012 & .093 & .044 & .157 \\
\midrule
\textbf{OSU-Mem}    & \textbf{.027} & \textbf{.130} & \textbf{.084} & \textbf{.221} \\
\ \ (disjoint)      & .008 & .075 & .034 & .138 \\
\bottomrule
\end{tabular}
\end{table}

\begin{figure}[h]
\centering
\includegraphics[width=.6\linewidth]{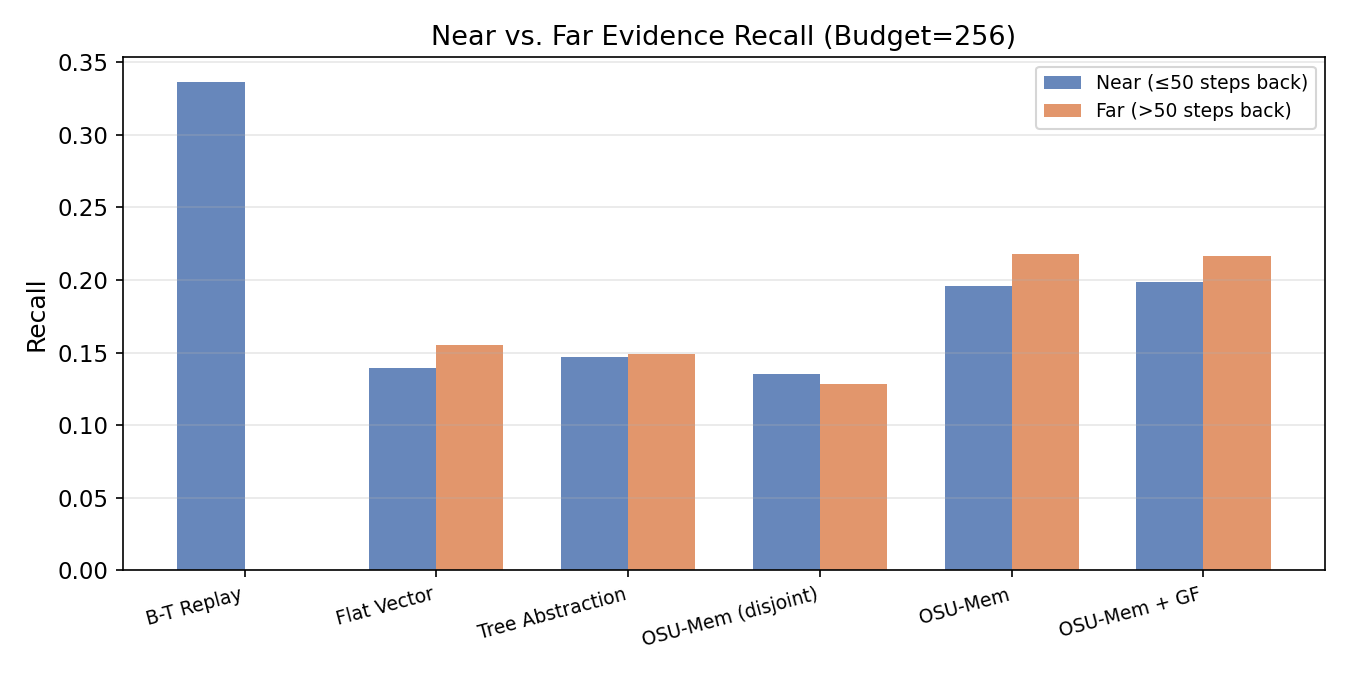}
\caption{Synthetic benchmark: near ($\le 50$ steps from query) vs.\ far ($> 50$ steps) evidence recall at $B{=}256$. Budget-Truncated Replay retrieves only near evidence (far recall $= 0$), illustrating recency bias. Flat Vector and Tree Abstraction show balanced near/far recall but at a lower level. OSU-Mem achieves the highest recall in \emph{both} regions, with a slight advantage on far evidence (Near $= .196$, Far $= .218$), confirming that its advantage is driven by non-local semantic structure rather than recency.}
\label{fig:near-far}
\end{figure}

\paragraph{$\tau$-bench (augmented cache).} To demonstrate OSU-Mem's strength under the T+E+ condition at scale, we apply a controlled multi-burst augmentation (Stage D, Appendix~\ref{app:taubench-augmented}) that lifts $100\%$ of queries into T+E+. On the augmented cache, OSU-Mem achieves the highest Hit@2 at every budget, with $+14.1$ pp over Flat Vector at $B{=}256$ (episode-level permutation $p{=}.001$), as illustrated by the budget-vs.-metric curves in Figure~\ref{fig:augmented-curves}. Full multi-budget results and paired permutation tests appear in Appendix~\ref{app:taubench-augmented}. The detailed argument that Stage D is not selectively friendly to OSU-Mem is given in Section~\ref{sec:multi-budget}.

\paragraph{ToolBench.} Because ToolBench's planted evidence is dense in the SBERT space, Flat Vector is a very strong baseline (nearly saturated: recall rises from $.686$ at $B{=}200$ to $.917$ at $B{=}1600$, 3-seed averaged), and the overlap variant matches Flat Vector within $\pm 1$ Recall pp. The primary statistical signal is therefore the \emph{overlap-vs.-disjoint} contrast, which isolates the T+E+ cell mechanism under a retrieval algorithm that fully exercises OSU structure. For transparency we note that OSU-Mem's \emph{default} coarse-to-fine algorithm performs poorly here (coarse-routing recall $<.17$; Section~\ref{sec:impl}), because all-MiniLM-L6-v2 entangles entity-, tool-, and semantic variation into mutually similar sibling centroids; this is why ToolBench uses the coverage-guided greedy variant, and it underscores that ToolBench supports the overlap \emph{construction principle} and not the full default OSU-Mem method. Table~\ref{tab:toolbench-main} reports overlap vs.\ disjoint under coverage retrieval at $\beta{=}.12$. Overlap beats disjoint at every budget, with Hit@2 gaps of $+6.2$, $+4.8$, $+3.7$, $+3.8$ pp and Recall gaps of $+2.7$, $+3.7$, $+2.1$, $+1.6$ pp at $B \in \{200, 400, 800, 1600\}$ respectively. The Hit@2 gap is largest at $B{=}200$ where budget pressure is tightest, and Recall gaps narrow toward higher budgets as disjoint approaches saturation.

\begin{table}[t]
\centering
\small
\caption{ToolBench planted-retrieval at $\beta{=}.12$, averaged over 3 seeds $\times$ 200 episodes. Coverage-guided retrieval with overlap ($M_{\max}{=}5$) vs.\ disjoint ($M_{\max}{=}1$). Both variants are within $\pm 1$ Recall pp of Flat Vector, and the structural signal is overlap $>$ disjoint.}
\label{tab:toolbench-main}
\resizebox{\linewidth}{!}{%
\begin{tabular}{lcccccccc}
\toprule
\textbf{Method} & \multicolumn{2}{c}{$B{=}200$} & \multicolumn{2}{c}{$B{=}400$} & \multicolumn{2}{c}{$B{=}800$} & \multicolumn{2}{c}{$B{=}1600$} \\
\cmidrule(lr){2-3}\cmidrule(lr){4-5}\cmidrule(lr){6-7}\cmidrule(lr){8-9}
 & Hit@2 & Recall & Hit@2 & Recall & Hit@2 & Recall & Hit@2 & Recall \\
\midrule
Flat Vector         & .427 & .686 & .583 & .792 & .698 & .863 & .808 & .917 \\
Flat $+$ MMR        & .228 & .567 & .420 & .704 & .560 & .797 & .733 & .885 \\
Disjoint Hierarchy  & .057 & .311 & .098 & .381 & .153 & .482 & .243 & .594 \\
\midrule
\textbf{Cov (overlap, $M{=}5$)} & \textbf{.430} & \textbf{.687} & \textbf{.588} & \textbf{.794} & \textbf{.702} & \textbf{.861} & \textbf{.792} & \textbf{.909} \\
Cov (disjoint, $M{=}1$)          & .368 & .659 & .540 & .757 & .665 & .840 & .753 & .894 \\
$\Delta$ overlap$-$disjoint      & $+$.062 & $+$.027 & $+$.048 & $+$.037 & $+$.037 & $+$.021 & $+$.038 & $+$.016 \\
\bottomrule
\end{tabular}}
\end{table}

\paragraph{Primary configuration.} We report $\beta{=}.12$ and evidence-size $3$ as the primary ToolBench configuration. $\beta{=}.12$ is the midpoint of the swept range $\{.04, .08, .12, .16, .20\}$---a conservative choice that leaves dose-response headroom on both sides rather than the maximum-effect endpoint---and evidence-size $3$ is the smallest value at which pairwise tool- and entity-sharing can hold beyond a single pair (evidence-size $2$ has only one pair). The full $\beta$ and evidence-size sweep, including the per-size recall gaps, appears in Appendix~\ref{app:toolbench-sensitivity}; the overlap advantage is directionally preserved across all three evidence sizes.

\paragraph{Direct confirmation from evidence-OSU statistics.} Averaged over 3 seeds, each with 200 episodes and 3 evidence steps per episode (600 evidence steps per seed, 1800 total), under overlap construction ($M_{\max}{=}5$) each evidence step belongs to $2.16$ OSUs on average (median $2$), whereas under disjoint construction ($M_{\max}{=}1$), the average drops to $.87$ (median $1$). Average disjoint membership \emph{below one} has a direct mechanistic interpretation: $13.2\%$ of evidence steps belong to \emph{no} disjoint OSU, because view-priority pruning at $M_{\max}{=}1$ forces each step into exactly one view and the minimum-size filter then deletes views with fewer than two members. Under overlap, $0\%$ of evidence is evicted. \textbf{This provides a controlled analogue of the $\tau$-bench T+E+ cell mechanism}: steps carrying evidence for multiple views lose their alternate retrieval handles in disjoint mode and gain them back in overlap mode. We read this overlap-vs.-disjoint gap as an \emph{upper bound} on the pure construction-principle effect: part of the disjoint disadvantage comes from the $f_{\min}{=}2$ minimum-size filter rather than from the absence of overlap alone. 

The three benchmarks provide complementary evidence at different levels of control: the synthetic benchmark validates the method in a fully controlled similarity landscape, the $\tau$-bench cross-tab identifies the cell-conditional effect on the $\tau$-bench-derived trajectories, and ToolBench confirms the overlap-construction mechanism on an independent data source with structural T+E+ verification. We develop this joint interpretation in Section~\ref{sec:discussion}.

\subsection{Evidence-Structure Cross-Tab Analysis}
\label{sec:cross-tab}

We ask a more precise question than whether OSU-Mem works on real data: \emph{on which type of real query does it work, and why?} For each of the 328 queries in the unaugmented $\tau$-bench cache we compute pairwise sharing of ground-truth evidence steps along three views: entity, tool, and subgoal. Only $14.3\%$ of queries have non-zero sharing across all three views, a regime the synthetic benchmark instantiates by construction for $39.3\%$ of queries (entity and subgoal sharing hold for $100\%$, but tool-signature sharing only for $39.3\%$). Entity sharing has median $0$, tool sharing has median $0$, and only subgoal sharing is substantive (median $.47$), because focal evidence usually lies within a single user-turn segment. The two views that OSU-Mem's multi-view constructor most depends on for \emph{binding scattered evidence into a single retrieval handle} (entity and tool) have almost nothing to bind on the majority of real $\tau$-bench queries.

We cross-tabulate the 328 queries by whether evidence steps share at least one tool signature (T$\pm$) and at least one entity (E$\pm$). Table~\ref{tab:cell-crosstab} reports the per-cell OSU-Mem$-$Flat Vector gap at $B{=}256$. The pattern is clearly bimodal. On T+E+ ($n{=}56$, $17\%$) OSU-Mem beats Flat Vector by $+17.9$ Hit@2 pp ($p{=}.022$), and on T$-$E$-$ ($n{=}132$, $40\%$) it is \emph{dominated} by $-12.9$ Hit@2 pp ($p{=}.007$) and $-15.7$ Recall pp ($p{<}.001$). The aggregate near-tie of OSU-Mem with Flat Vector on the full cache is therefore a \emph{consequence of cell mixture} rather than a property of either method. This bimodal pattern is consistent with the prediction of Propositions~\ref{prop:reachability}--\ref{prop:dilution}: overlap increases reachability on T+E+ evidence while centroid dilution degrades performance on T$-$E$-$ evidence.

\begin{table}[t]
\centering
\small
\caption{$\tau$-bench $2{\times}2$ cross-tabulation at $B{=}256$: unaugmented cache, 328 queries. Paired sign-flip permutation (10,000 permutations) on per-query differences. Stars: $^{.}p{<}.10$, $^{*}p{<}.05$, $^{**}p{<}.01$, $^{***}p{<}.001$.}
\label{tab:cell-crosstab}
\begin{tabular}{lcccc}
\toprule
\textbf{Cell} & $n$ & \textbf{OSU$-$Flat $\Delta$Hit@2} & \textbf{OSU$-$Flat $\Delta$Recall} & \textbf{overlap$-$disjoint $\Delta$Hit@2} \\
\midrule
T+E+ (both share)      & 56  & $+.179^{*}$        & $+.080^{*}$        & $+.089$ \\
T+E$-$ (tool only)     & 96  & $-.031$            & $-.047^{.}$        & $+.073$ \\
T$-$E+ (entity only)   & 44  & $+.159$            & $+.000$            & $-.068$ \\
T$-$E$-$ (neither)     & 132 & $-.129^{**}$       & $-.157^{***}$      & $-.083^{*}$ \\
\midrule
\textbf{Aggregate}     & 328 & $-.010$            & $-.063$            & $-.006$ \\
\bottomrule
\end{tabular}
\end{table}

\subsection{Multi-Budget Stability and Interaction Tests}
\label{sec:multi-budget}

\paragraph{Multi-budget stability of the T+E+ advantage.} Re-running the cross-tab at all four budgets (Table~\ref{tab:cell-crosstab-multi}) shows a stable T+E+ advantage---$+.161$, $+.232$, $+.179$, $+.161$ Hit@2 pp at $B \in \{64,128,256,512\}$, each significant at $p<.05$---alongside a T$-$E$-$ disadvantage that grows monotonically with budget, from $+.008$ at $B{=}64$ to $-.364$ at $B{=}512$ ($p<.001$). Because all four budgets score the same 56 T+E+ queries, this is stability of sign and effect size rather than independent replication; with $n{=}56$, the claim rests on this stability together with the independent ToolBench corroboration (Section~\ref{sec:overall}) rather than on any single per-budget $p$-value. Under a deliberately weaker query construction (a single random evidence step plus noise), the sign holds and ${\approx}80\%$ of the $B{=}256$ magnitude survives ($+.179 \to +.143$), but single-cell significance softens to marginal ($p{:}\,.022 \to .094$). A natural mechanistic reading is that at small budgets structured selection must make hard choices and OSU-Mem's multi-handle index helps any cross-view-bound query, whereas at large budgets Flat Vector can recover most evidence by cosine alone and OSU-Mem's structural overhead on unbound evidence becomes pure cost.

\begin{table}[t]
\centering
\small
\caption{$\tau$-bench: OSU-Mem$-$Flat Vector Hit@2 across all four budgets by $2{\times}2$ cell. The T+E+ advantage is stable across budgets, and the T$-$E$-$ disadvantage grows monotonically with budget. Stars as in Table~\ref{tab:cell-crosstab}.}
\label{tab:cell-crosstab-multi}
\begin{tabular}{lcccc}
\toprule
\textbf{Cell} & $B{=}64$ & $B{=}128$ & $B{=}256$ & $B{=}512$ \\
\midrule
T+E+ ($n{=}56$)     & $+.161^{*}$  & $+.232^{**}$ & $+.179^{*}$   & $+.161^{*}$ \\
T+E$-$ ($n{=}96$)   & $+.094^{.}$  & $+.125^{.}$  & $-.031$       & $-.146^{**}$ \\
T$-$E+ ($n{=}44$)   & $+.045$      & $+.114$      & $+.159$       & $+.227^{*}$ \\
T$-$E$-$ ($n{=}132$) & $+.008$     & $-.023$      & $-.129^{**}$  & $-.364^{***}$ \\
\bottomrule
\end{tabular}
\end{table}

\paragraph{Method $\times$ cell interaction test.} The previous tables report per-cell marginal effects. The statistically correct tool for our core heterogeneous-treatment claim is the \emph{method $\times$ cell interaction test}: under the null that the OSU$-$Flat gap is identical across the four cells, what is the probability of observing a spread as large as we do? We compute this via a paired sign-flip permutation on the per-query (OSU$-$Flat) differences after mean-centering within each cell, so that the test isolates between-cell variance in the effect. The interaction is strongly significant at our primary budget $B{=}256$ (Hit@2 $p{=}.002$, Recall $p{=}.021$), marginal at $B{=}128$ ($p{=}.081$ Hit@2), and null at $B{=}64$ and $B{=}512$ ($p{=}.308$ and $p{=}.558$ Hit@2). At the extremes this is expected. At $B{=}64$ all four cells show positive effects (range $+.008$ to $+.161$), so the between-cell variance in the OSU$-$Flat gap is small relative to within-cell noise, and at $B{=}512$ the T+E+ and T$-$E+ cells both show large positive effects while T+E$-$ and T$-$E$-$ both show negatives, again reducing the specific T+E+-vs.-others contrast that our permutation permutes. The interaction test therefore pinpoints $B{=}256$ as the primary budget at which the cell-conditional claim is the most clearly isolated.

\paragraph{Cross-view binding: T+E+ is the strongest but not the only positive-binding cell.} The T$-$E+ cell (entity-only sharing, $n{=}44$) shows a \emph{monotonically growing} Hit@2 advantage of $+.045$, $+.114$, $+.159$, $+.227^{*}$ across $B \in \{64,128,256,512\}$, significant at $B{=}512$ and directionally aligned with T+E+ at every budget. The T+E$-$ cell (tool-only sharing) is positive at small budgets and turns negative at large ones. The broader reading is therefore ``OSU-Mem helps when \emph{at least one} pairwise-sharing view (entity or tool) binds evidence'' rather than the stricter ``only when both bind.'' T+E+ is the strongest signal because both views reinforce, and T$-$E+ is the next most distinct case because entity sharing alone gives the entity-view OSU a handle that Flat Vector lacks. This matters for the predictive rule (Section~\ref{sec:discussion}): a rule based on the tool-sharing marginal $\pi_\text{T+}$ alone under-counts regimes where OSU-Mem helps, and we report an alternative ``any-binding'' marginal in Appendix~\ref{app:deployment}.

\paragraph{Refined mechanism and regime-lift probe.} Taken together, the cell analysis yields a refined claim: OSU-Mem's overlap and multi-view design \emph{pay off specifically when evidence is bound together by shared tool signatures or shared entities}. To verify at scale (and avoid the $n{=}56$ T+E+ power limit), we perform a controlled multi-burst augmentation. Each focal evidence step is replicated as 3 scattered copies inheriting the anchor's entity, tool, and subgoal annotations. This lifts $100\%$ of queries into T+E+ (up from $17.1\%$). The T+E+ structure here is a \emph{mechanical consequence} of metadata inheritance with small embedding noise rather than a tuned parameter, and the predicted OSU-vs.-Flat effect is derived from the T+E+ cross-tab row before Stage D is implemented. Crucially, Stage~D benefits \emph{any} cosine-retrieval method, not OSU-Mem selectively: Flat Vector, Flat\,$+$\,MMR, and Coverage-Aware all improve substantially on the augmented cache (e.g., Flat Vector Hit@2 at $B{=}256$ rises from $.045$ to $.328$ query-level). What survives is the \emph{relative} OSU$-$Flat gap ($+.107$ Hit@2 at $B{=}256$), of the same order as the unaugmented T+E+-cell gap ($+.179$); this rules out the reading that the T+E+ effect is incidental to $\tau$-bench's mixture and would vanish under a controlled lift. The per-method improvement breakdown is in Appendix~\ref{app:taubench-augmented}. On the augmented cache, OSU-Mem achieves Hit@2 $.435$ vs.\ Flat $.328$ at $B{=}256$ ($\Delta{=}{+}.107$ query-level; episode-level permutation $\Delta{=}{+}.141$, $p{=}.001$), outperforming Flat at every budget. Full multi-budget tests, augmented-cache ablations, and a separable finding at $B{=}64$ ($2.0$ to $2.4{\times}$ Hit@2 over flat baselines at matched Recall) appear in Appendix~\ref{app:taubench-augmented}. SBERT efficiency numbers (where OSU-Mem's $6{\times}$ dimensional cost changes its ordering with flat baselines) appear in Appendix~\ref{app:taubench-eff}.

\begin{figure}[h]
\centering
\includegraphics[width=\linewidth]{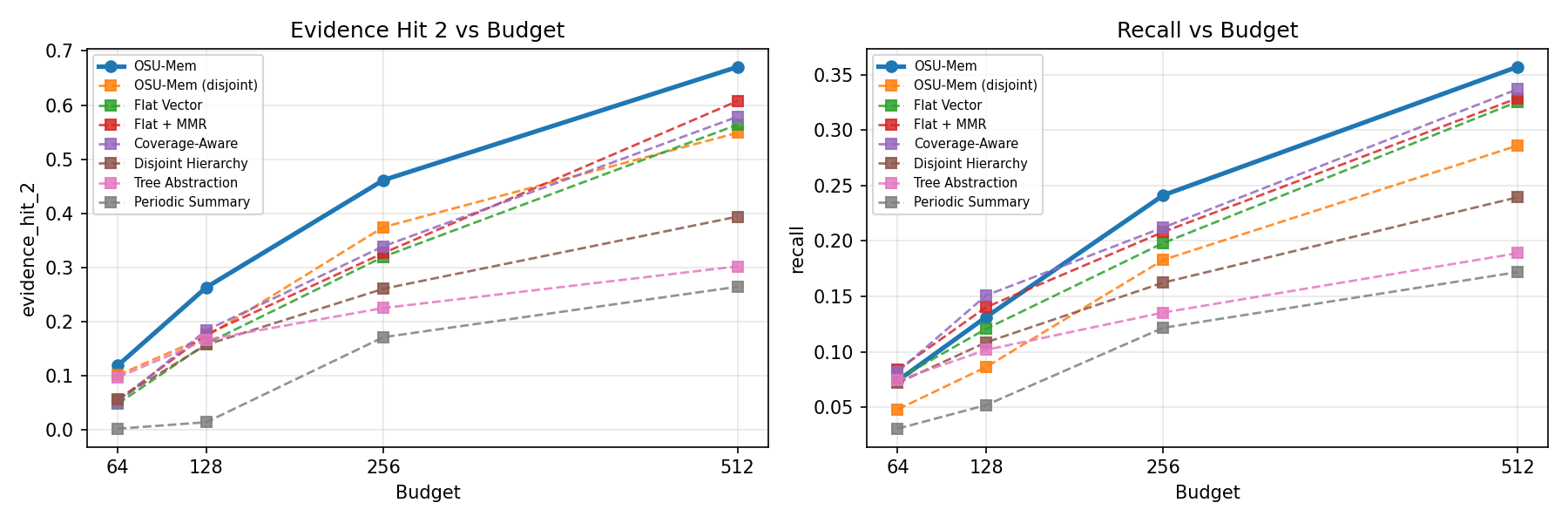}
\caption{$\tau$-bench multi-burst-augmented cache (125 episodes, 328 queries): Hit@2 (left) and Recall (right) vs.\ token budget, episode-level averaging. OSU-Mem (solid blue) achieves the highest Hit@2 at every budget. All retrieval methods improve substantially under augmentation relative to the unaugmented cache (e.g., Flat Vector Hit@2 rises from $.045$ to $.320$ at $B{=}256$, both episode-level), confirming that Stage~D is not selectively friendly to OSU-Mem. The persistent gap between OSU-Mem and Flat Vector at every budget is the relative advantage that the cross-tab's T+E+ cell analysis predicts.}
\label{fig:augmented-curves}
\end{figure}

\subsection{Ablation Studies}
\label{sec:ablation}

\paragraph{Effect of overlap and multi-view construction (synthetic).} At $B{=}256$, removing overlap ($\to$ disjoint OSUs) drops Hit@2 from $.084$ to $.034$ ($-59.6\%$) and Recall from $.221$ to $.138$ ($-37.4\%$), which is the largest single-factor ablation effect on this benchmark. Removing multi-view construction (sim-only OSUs) drops to Hit@2 $.014$ and Recall $.110$, below all competitive retrieval baselines, which confirms that the union of entity, tool, subgoal, and similarity views is critical. The novelty penalty has a small effect that is not statistically significant ($p{=}.307$ for Recall, $p{=}.876$ for Hit@2 at $B{=}256$). The novelty term is part of the default OSU-Mem method but is budget-gated (active only when $B \ge 192$), and its measured contribution is dominated by the overlap and multi-view construction effects. An $M_{\max}$ sweep saturates at $M_{\max}{=}5$ (one slot per view plus one slack slot). Full ablation flag combinations appear in Appendix~\ref{app:synth-ablation}.

\paragraph{Ablation on $\tau$-bench augmented cache.} Consistent with the synthetic findings, disjoint OSU-Mem achieves Hit@2 $.387$ vs.\ overlap's $.435$ at $B{=}256$ (query-level); the paired episode-level permutation test confirms significance ($p{=}.012$, Appendix Table~\ref{tab:taubench-perm}). Similarity-only OSUs drop to $.285$. The overlap-vs.-disjoint gap activates at $B \ge 128$ and the effect sizes against hierarchical baselines (Disjoint Hierarchy, Tree Abstraction) grow monotonically with budget, reaching $+.28$ and $+.37$ Hit@2 pp at $B{=}512$. OSU-Mem $+$ GF (ground-truth filtering) matches the default OSU-Mem at $B \le 256$ and gains only $+.004$ Hit@2 at $B{=}512$, indicating that on the augmented cache the coarse stage already selects the correct OSUs with high probability, leaving little room for oracle guidance. Full augmented-cache results appear in Appendix~\ref{app:taubench-augmented}.

\subsection{Dose-Response in the Coverage Bonus}
\label{sec:dose-response}

The decisive check for ToolBench is whether the overlap advantage shows a dose-response in $\beta$. If the effect is driven by OSU structure, the gap should grow as the retrieval algorithm's reliance on that structure grows. Table~\ref{tab:toolbench-beta-sweep} reports the full $\beta$ sweep on evidence recall (the metric with per-seed Fisher-combined permutation tests). The gap grows \emph{monotonically} with $\beta$ at every budget, as visualized in Figure~\ref{fig:beta-sweep}. At $B{=}400$, it grows from $+1.2$ pp at $\beta{=}.04$ to $+6.4$ pp at $\beta{=}.20$. Hit@2 shows the same pattern with larger absolute effects, reaching $+10.8$ pp at $\beta{=}.20$ (Appendix~\ref{app:toolbench-sensitivity}). Directional consistency across $\beta \times B \times \text{seed}$ cells is $57/60 = 95\%$, with Fisher-combined $p < .001$ at $B{=}200$ and $B{=}400$ and $p < .05$ at every budget for $\beta \ge .12$. \textbf{Both the monotonic-in-$\beta$ dose-response and the structural T+E+ argument substantially reduce the plausibility of alternative explanations} (incidental differences in retrieval pool composition, SBERT embedding artifacts, ToolBench-specific artifacts). Overlap provides a specific additional retrieval pathway whose value scales with how much the retrieval algorithm rewards OSU coverage.

\begin{table}[t]
\centering
\small
\caption{ToolBench: overlap$-$disjoint \emph{Recall} gap (pp) across $\beta$ and $B$, averaged over 3 seeds $\times$ 200 episodes. Monotonic dose-response in $\beta$ at every budget. Stars: $^{*}p{<}.05$, $^{**}p{<}.01$, $^{***}p{<}.001$ Fisher-combined across seeds. Hit@2 version in Appendix~\ref{app:toolbench-sensitivity}.}
\label{tab:toolbench-beta-sweep}
\begin{tabular}{lcccc}
\toprule
$\beta$ & $B{=}200$ & $B{=}400$ & $B{=}800$ & $B{=}1600$ \\
\midrule
$.04$          & $+.5$               & $+1.2^{**}$           & $+.3$                & $+.2$              \\
$.08$          & $+1.3^{*}$           & $+2.6^{***}$          & $+.8$                & $+.9$              \\
$\mathbf{.12}$ & $\mathbf{+2.7^{***}}$ & $\mathbf{+3.7^{***}}$ & $\mathbf{+2.1^{**}}$  & $\mathbf{+1.6^{*}}$ \\
$.16$          & $+4.2^{***}$         & $+4.7^{***}$          & $+4.0^{***}$          & $+2.6^{***}$        \\
$.20$          & $+5.1^{***}$         & $+6.4^{***}$          & $+5.4^{***}$          & $+4.2^{***}$        \\
\bottomrule
\end{tabular}
\end{table}

\begin{figure}[h]
\centering
\includegraphics[width=.6\linewidth]{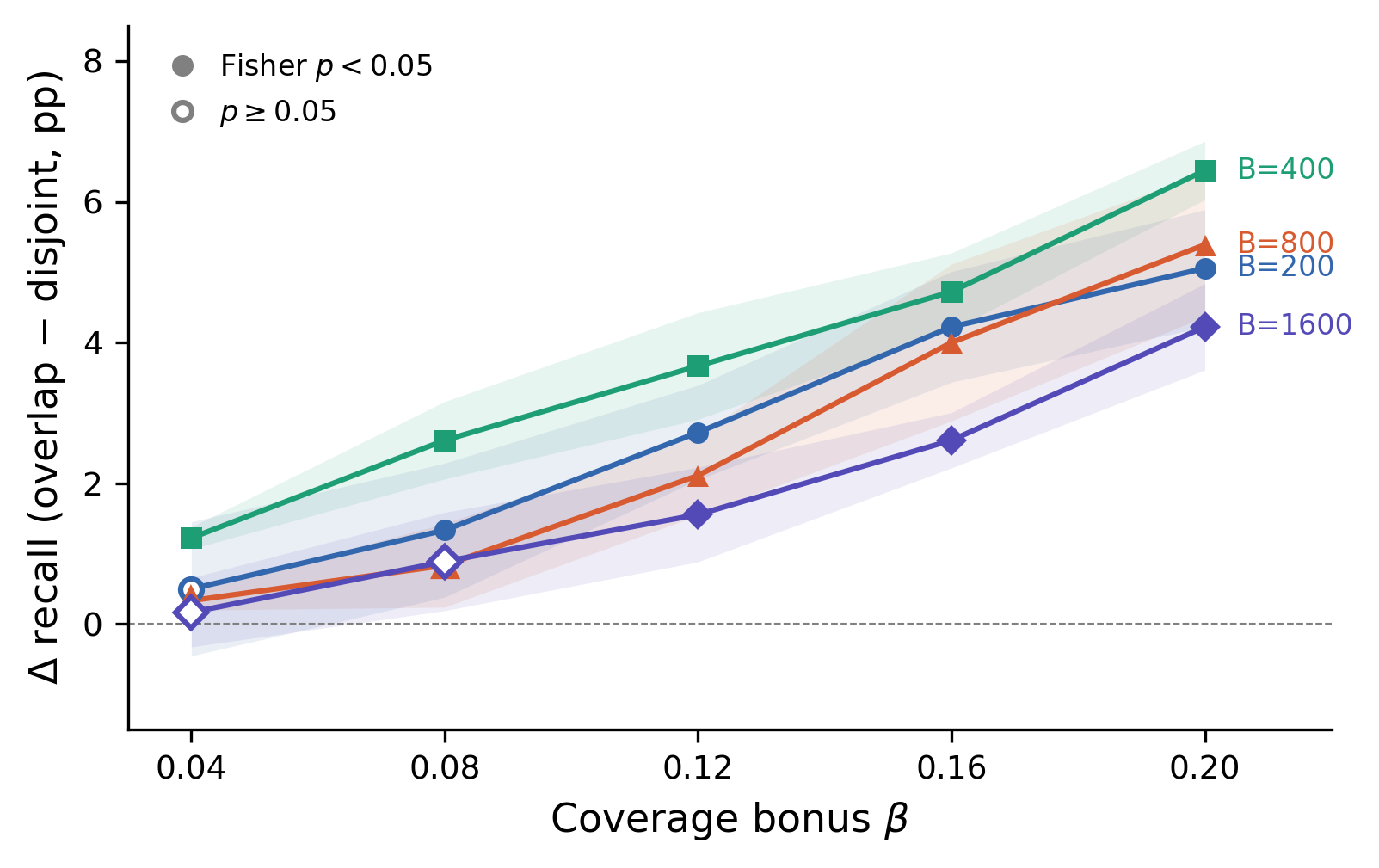}
\caption{ToolBench dose-response: overlap$-$disjoint Recall gap (pp) as a function of coverage bonus $\beta$, at four budgets. Each point averages 3 seeds $\times$ 200 episodes. Filled markers indicate Fisher-combined $p < .05$; open markers indicate $p \ge .05$. The monotonic increase at every budget confirms that the overlap advantage scales with how much the retrieval algorithm rewards OSU coverage, consistent with the structural T+E+ mechanism rather than an incidental artifact. Shaded bands show the per-seed range.}
\label{fig:beta-sweep}
\end{figure}

A secondary negative result is informative: a cohesion-style greedy variant that penalizes similarity to already-selected steps ($-\alpha \cdot \max_{j \in \mathcal{S}} \cos(\mathbf{h}_i, \mathbf{h}_j)$) does \emph{not} produce an overlap advantage at any $\alpha$ or budget (Appendix~\ref{app:toolbench-cohesion-null}). This indicates that the T+E+ mechanism depends on the retrieval algorithm actively \emph{rewarding OSU coverage} rather than on merely discouraging within-pool redundancy.

\subsection{LLM-Mediated Evidence Selection}
\label{sec:llm-e2e}

As an additional check, we ask whether the cross-tab's retrieval quality differences are preserved when the retrieved context is consumed by an LLM. For each of $414$ queries on the unaugmented cache (cell distribution $25.8/31.6/10.6/31.9\%$ T+E+/T+E-/T-E+/T-E-, which \emph{differs} from the $n{=}328$ cross-tab in Section~\ref{sec:cross-tab} because the LLM experiment uses a broader inclusion rule that relaxes the focal-source filter, detailed in Appendix~\ref{app:e2e-discussion}), we run OSU-Mem, Flat Vector, and no-memory at $B{=}256$, feed the retrieved context with the user instruction to GPT-4o-mini ($T{=}0$), ask for exactly three short descriptions of the most relevant past events, SBERT-re-encode the listed events, and score against ground-truth evidence embeddings via \textbf{soft recall} (mean max cosine) and \textbf{hard recall} (fraction with max cosine $> .5$). Three seeds (Python random state; retrieval and OSU construction are fully deterministic, so residual per-seed variance arises from non-determinism in the OpenAI inference backend at $T{=}0$), paired permutation after per-query averaging (the initial pooled-seed test was anti-conservative since $T{=}0$ makes seeds near-duplicates, see Appendix~\ref{app:e2e-discussion}).

\begin{table}[t]
\centering
\small
\caption{LLM-mediated evidence selection on $\tau$-bench unaugmented cache, $B{=}256$, GPT-4o-mini, 3 seeds, paired permutation by query (after seed averaging), $20{,}000$ permutations. Bracketed ranges are 95\% confidence intervals (CIs). Stars: $^{*}p{<}.05$. No-memory floor at hard recall $.003$ confirms the LLM uses retrieved context.}
\label{tab:taubench-e2e}
\resizebox{\linewidth}{!}{%
\begin{tabular}{lccccccc}
\toprule
 & & \multicolumn{3}{c}{\textbf{Hard recall (mean)}} & & \multicolumn{2}{c}{\textbf{$\Delta$ OSU$-$Flat (95\% CI)}} \\
\cmidrule(lr){3-5}\cmidrule(lr){7-8}
Cell & $n$ & OSU & Flat & None & & Soft & Hard \\
\midrule
T+E+        & 107 & $\mathbf{.230}$ & $.166$ & $.008$ & & $+.005$ {\scriptsize $[-.017,+.027]$}        & $\mathbf{+.064^{*}}$ {\scriptsize $[+.008,+.121]$} \\
T+E$-$      & 131 & $.169$          & $.146$ & $.000$ & & $+.019$ {\scriptsize $[-.001,+.039]$}        & $+.023$ {\scriptsize $[-.030,+.075]$}              \\
T$-$E+      & 44  & $.252$          & $.213$ & $.008$ & & $+.016$ {\scriptsize $[-.018,+.051]$}        & $+.039$ {\scriptsize $[-.039,+.123]$}              \\
T$-$E$-$    & 132 & $.074$          & $.066$ & $.000$ & & $+.000$ {\scriptsize $[-.016,+.017]$}        & $+.008$ {\scriptsize $[-.026,+.041]$}              \\
\midrule
\textit{Aggregate} & 414 & $.164$ & $.133$ & $.003$ & & $+.009$ & $+.031$ \\
\bottomrule
\end{tabular}}
\end{table}

\paragraph{Findings.} The T+E+ advantage survives the LLM stage: $+6.4$ pp hard recall ($p{=}.025$), the only cell individually significant after correction (per-cell CIs in Table~\ref{tab:taubench-e2e}), while T$-$E$-$ is null ($p_\text{hard}{=}.65$) exactly as the cross-tab predicts. Two bounds apply: this is an evidence-\emph{selection} check, not a full task evaluation (the LLM produces no tool calls and takes no actions), and it shares the SBERT encoder---and hence some variance---with the cross-tab of Section~\ref{sec:cross-tab}, so it is a consistency check rather than independent corroboration. The pooled-vs.-corrected comparison, the near-circularity discussion, and the soft-vs.-hard mechanism interpretation are in Appendix~\ref{app:e2e-discussion}.

\section{Discussion}
\label{sec:discussion}

\paragraph{Two datasets, complementary roles.} The $\tau$-bench cell analysis and the ToolBench $\beta$ sweep are not two interchangeable corroborations of the same claim; they play \emph{asymmetric, complementary} roles. $\tau$-bench is a naturally mixed distribution and serves to \emph{measure the heterogeneity} of the effect across cells (via the cross-tab), including the cells where overlap hurts. ToolBench artificially restricts the distribution to the favorable cell and serves to \emph{isolate the mechanism} under a known-favorable structure on an independent source. They differ in three ways, each of which is informative. \textbf{(1) Evidence distribution and role.} $\tau$-bench has a naturally mixed distribution ($17\%$ T+E+, $40\%$ T$-$E$-$, $43\%$ mixed) and is the only analysis that measures cross-cell heterogeneity; ToolBench is structurally T+E+ by construction and isolates the favorable-cell mechanism. \textbf{(2) Retrieval algorithm.} $\tau$-bench uses the full OSU-Mem coarse-to-fine procedure, and ToolBench uses the coverage-guided greedy variant. Agreement across algorithms shifts credit from the retrieval procedure onto the \emph{construction principle}. \textbf{(3) Null comparison.} $\tau$-bench's T+E+ internal overlap-vs.-disjoint gap is $+8.9$ Hit@2 pp at $B{=}256$, and ToolBench's overlap-vs.-disjoint gap is $+3.7$ to $+6.2$ Hit@2 pp across all four tested budgets (at $\beta{=}.12$). The magnitudes agree to within a factor of two. 
A sub-population effect measured on $\tau$-bench alone could reflect incidental properties of a single data source; ToolBench guards against this by reproducing the same mechanism signature on an independent source with code-level structural verification of cell assignment. Conversely, ToolBench alone cannot speak to heterogeneity across cells, and $\tau$-bench's per-cell cross-tab is the only analysis that measures this. The two datasets together are stronger than either alone.

\paragraph{A predictive rule for when overlap helps (derived).} The joint cell analysis yields a quantitative heuristic whose thresholds are derived directly from the $\tau$-bench cross-tab; the full linear predictor, break-even points, and budget table appear in Appendix~\ref{app:deployment}. We stress that the rule predicts a \emph{retrieval-quality} outcome (the aggregate Hit@2 gap over a query population), not end-to-end agent task success, which we do not measure (Limitation~iv); it is therefore a hypothesis about where overlap is worth applying, to be confirmed in closed-loop evaluation. Let $\pi_\text{T+}$ denote the expected fraction of queries with at least one pairwise tool-sharing edge among expected evidence. This can be approximated from the agent framework's tool-call logs without running retrieval, by using the fraction of recent episodes that invoke the same tool signature at least twice as a proxy for expected evidence-step tool sharing; when tool-call logs are unavailable, an entity-sharing rate estimated from named entity recognition (NER) over recent trajectories can substitute. The qualitative rule is simple: OSU-Mem is likely net-positive on aggregate retrieval Hit@2 when $\pi_\text{T+}$ is high and likely net-negative when it is low, with a held-out calibration recommended in the intermediate band. Two qualitative dependencies matter when applying the rule. First, the break-even threshold \emph{rises with budget}: at small budgets any cross-view binding helps, whereas at large budgets the T$-$E$-$ disadvantage grows and demands a higher $\pi_\text{T+}$ to remain net-positive. Second, aggregate Recall requires a stricter threshold than Hit@2, because the T$-$E$-$ Recall disadvantage is the larger of the two. We also recommend source-level namespace prefixing as a standard precondition whenever trajectories from multiple sources are concatenated, to avoid the cross-source metadata contamination pitfall (Appendix~\ref{app:taubench-pipeline}).

\paragraph{How much structure does the aggregate near-tie hide?} The cross-tab lets us quantify it directly. Routing each query to its better method on its evidence cell---OSU-Mem for the favorable cells (T+E+, T$-$E+) and Flat Vector otherwise---turns the aggregate Hit@2 near-tie into ${\approx}+5.2$ pp over Flat Vector and lifts aggregate Recall from $-6.3$ to ${\approx}+1.4$ pp. This is the cell-mixture artifact, this paper's central claim, made quantitative: the aggregate null conceals ${\sim}5$ Hit@2 pp of recoverable structure. The figure is an oracle-routed upper bound. This is an in-sample oracle ceiling, not an evaluated or deployable routing method. The per-cell arithmetic, the label-free (metadata-based) policy that approximates it at decision time, and the future-work routing are in Appendix~\ref{app:deployment}.

\paragraph{Limitations.} \textbf{(i)} OSU-Mem depends on embedding quality, and our $\tau$-bench pipeline replaces focal query embeddings with noisy means of evidence; ToolBench's use of verbatim natural-language queries (Section~\ref{sec:overall}) is our strongest evidence that it is not the dominant driver of the T+E+ advantage. \textbf{(ii)} In production, OSU construction relies on extraction heuristics (NER, pattern matching, intent-shift detection) whose errors are not captured by our benchmarks. \textbf{(iii)} Efficiency claims do not transfer unchanged from 64-d to 384-d embeddings. On SBERT, OSU-Mem has the highest retrieve-only latency and the largest storage footprint (${\approx}2.6$\,ms retrieve-only; small relative to typical LLM-call latency, but whether this overhead is acceptable should be validated in the target agent loop, see Appendix~\ref{app:taubench-eff}). \textbf{(iv)} We validate retrieval quality and an LLM-mediated evidence-selection step (Section~\ref{sec:llm-e2e}), but not full multi-step task success in an interactive environment with an action loop and reward. Because the recall-to-task-success relationship is task-dependent \citep{petroni-etal-2021-kilt}, the full task-level impact requires benchmark-specific end-to-end evaluation, which we leave to future work. \textbf{(v)} No public benchmark known to us simultaneously has long trajectories ($\ge 150$ steps) \emph{and} a naturally T+E+ evidence distribution in a single unaugmented dataset. $\tau$-bench gives the length but not the cell concentration, and ToolBench gives the cell concentration but requires planted-retrieval to reach comparable lengths. Building such a benchmark is an actionable direction for future work.

\section{Conclusion}
\label{sec:conclusion}

OSU-Mem is a trajectory memory framework for long-horizon agents that organizes interaction histories into overlapping semantic units and retrieves via budgeted, query-adaptive coarse-to-fine expansion. On two trajectory-derived benchmarks---constructed retrieval settings based on $\tau$-bench and ToolBench---we characterize the precise condition under which the method helps. In a synthetic setting that is E+ by construction (entity-shared for $100\%$ of queries; tool-signature-shared, hence strictly T+E+, for $39.3\%$), it improves recall by $+39.9\%$ and Hit@2 by $+61.5\%$ over the strongest baseline, confirming the mechanism in the regime the theory predicts. The synthetic and $\tau$-bench experiments support the full OSU-Mem coarse-to-fine pipeline, while the ToolBench experiment supports the overlap construction principle under a coverage-guided greedy variant. OSU-Mem is strong and significant on T+E+ evidence (where steps share tool signatures and entities across pairs) and weak on fully heterogeneous evidence, with favorable-cell prevalence estimable from low-cost metadata proxies. $\tau$-bench measures the per-cell effect on naturally mixed evidence via cross-tab, and ToolBench measures it on a distribution structurally restricted to the favorable cell by construction. The mechanism signatures agree across two datasets, two retrieval algorithms, and two ways of instantiating the T+E+ condition. We note that the $\tau$-bench evaluation relies on a constructed query-evidence alignment (Section~\ref{sec:datasets}) rather than natural user queries, and that neither benchmark constitutes a closed-loop agent task evaluation (Section~\ref{sec:discussion}, Limitation~iv). We view OSU-Mem as a design point whose conditional value is now characterized on trajectory-derived retrieval settings, rather than generically claimed on deployment distributions.

\bibliographystyle{plainnat}
\bibliography{references}
\newpage
\appendix
\section{Extended Method Details}
\label{app:method}

\paragraph{Hyperparameter defaults.} Table~\ref{tab:hyperparams} lists all OSU-Mem hyperparameters, their default values, and the rationale for each choice. All values are fixed across the three benchmarks (synthetic, $\tau$-bench, ToolBench); none was tuned per-benchmark, and in particular none was tuned on the $\tau$-bench T+E+ cell (see Section~\ref{sec:method-hyperparams}).

\begin{table}[!htbp]
\centering
\small
\caption{OSU-Mem hyperparameter defaults. All values are fixed across the three benchmarks.}
\label{tab:hyperparams}
\begin{tabular}{@{}>{\raggedright\arraybackslash}p{2.4cm}>{\raggedright\arraybackslash}p{4.0cm}>{\raggedright\arraybackslash}p{6.8cm}@{}}
\toprule
Name & Default & Rationale \\
\midrule
$\Delta_{\text{gap}}$ & $8$ & Base temporal-gap splitting threshold (step-index window). The three view-specific thresholds below are derived from this value. \\
$\Delta_{\text{ent}}$ & $\max(60, 8\Delta_{\text{gap}}) = 64$ & Entities persist across long segments (account numbers, destinations), so we do not want to split entity OSUs at short temporal gaps. \\
$\Delta_{\text{tool}}$ & $\max(12, 2\Delta_{\text{gap}}) = 16$ & Tool-signature recurrences are more localized, and a long gap between two uses of the same tool is usually a semantic boundary. \\
$\Delta_{\text{sg}}$ & $\max(8, \Delta_{\text{gap}}) = 8$ & Subgoals are the shortest semantic unit and split most aggressively. \\
$S_{\max}$ & $20$ & Upper bound on OSU size, preventing degenerate all-steps OSUs on long trajectories. \\
$f_{\min}$ & $2$ & Minimum OSU size. OSUs with a single member are pruned (they have no retrieval advantage over Flat Vector). \\
$M_{\max}$ & $5$ & Cap on OSU memberships per step. $5 = $ (1 entity $+$ 1 tool $+$ 1 subgoal $+$ 1 similarity $+$ 1 slack). Ablation shows saturation at $M_{\max}{=}5$. \\
$K_{\text{coarse}}$ & $\min(|\{u_k\}|,$ \newline $\max(8,\; |\{u_k\}|/10,\; B/30))$ & Adaptive: scan roughly 10\% of the OSU pool at the coarse stage, with a floor of 8 and a budget-proportional term $B/30$ that ensures enough OSUs are scored when the pool is small but the budget is large. \\
$\beta$ (coverage bonus) & $.12$ & ToolBench-specific, chosen as the midpoint of the $\{.04,\ldots,.20\}$ sweep and reported as the primary $\beta$ in tables. Full $\beta$ sweep in Table~\ref{tab:toolbench-beta-sweep}. \\
$\lambda_{\mathrm{nov}}$ (novelty scale) & $.3$ & Scales the blended novelty contribution in the cost-normalized OSU scoring formula. Part of the default method, active only when $B \ge B_{\mathrm{nov}}$. Low-sensitivity parameter. \\
$B_{\mathrm{nov}}$ (novelty budget gate) & $192$ & Budget threshold below which novelty scoring is disabled. At small budgets the candidate pool is too small for novelty to provide useful signal. \\
$\alpha$ (novelty blend weight) & $.4$ & Interpolates embedding novelty and structural (entity or tool) novelty within the blended novelty term. Part of the default method (budget-gated). Low-sensitivity parameter (ablation shows $\pm .2$ moves Hit@2 by $<$.01). \\
Rebuild period $R$ & full-batch & OSU pool is rebuilt on every episode start. Online incremental rebuild is a future-work item. \\
View priority & entity $>$ tool $>$ similarity $>$ subgoal & When $M_{\max}$ is exceeded, memberships in lower-priority views are pruned. Entity is prioritized because entity annotations are the most semantically stable cross-step identifiers. $\tau$-bench's cell analysis subsequently confirms the entity-first choice. \\
\bottomrule
\end{tabular}
\end{table}

\paragraph{Novelty term (default, budget-gated).} The default OSU-Mem retrieval includes a budget-gated novelty-aware OSU scoring that penalises redundancy against already-selected steps. For each candidate step $i$ in an expanding OSU, novelty is scored against the selected set $\mathcal{S}$ as
\begin{align*}
\textsc{EmbNov}(i; \mathcal{S}) &= 1 - \max_{j\in\mathcal{S}}\cos(\mathbf{h}_i, \mathbf{h}_j), \\
\textsc{StructNov}(i; \mathcal{S}) &= \min\bigl(1,\; .6|\textsc{Ent}(i)\setminus\textsc{Ent}(\mathcal{S})| + .4\mathbb{1}[\textsc{Tool}(i)\notin\textsc{Tool}(\mathcal{S})]\bigr), \\
\textsc{Nov}(i;\mathcal{S}) &= (1-\alpha)\textsc{EmbNov} + \alpha\textsc{StructNov},
\end{align*}
with $\alpha = .4$. In the default method, the per-candidate novelty scores are averaged across an OSU's candidates and combined with relevance at the OSU level via the cost-normalized scoring formula $(r_k + \lambda_{\mathrm{nov}} \cdot \nu_k) / (\hat{c}_k + \epsilon)$ with $\lambda_{\mathrm{nov}} = .3$. The novelty term is budget-gated: it activates only when $B \ge B_{\mathrm{nov}} = 192$. \textbf{Its measured effect is small and statistically insignificant on every benchmark we tested} (synthetic $p_\text{Recall}{=}.307$, $p_\text{Hit@2}{=}.876$, $\tau$-bench augmented similarly, see Appendix~\ref{app:synth-ablation} and Appendix~\ref{app:taubench-augmented}). The dominant retrieval benefit derives from the overlap and multi-view construction rather than from the novelty term. We retain the ablation variant ``OSU-Mem (no novelty)'' in our tables so that readers can verify this null finding.

\paragraph{Deterministic packing and complexity.} Selected steps are packed in descending step--query similarity order before OSU summaries (summary-first packing is a robustness check, see Appendix~\ref{app:synth-ablation}). OSU rebuild period $R$ trades update fidelity against overhead. Online retrieval is dominated by a linear scan over OSU centroids ($O(|\text{OSUs}|)$) plus bounded candidate scoring, and an approximate nearest neighbor (ANN) index over centroids would make the coarse stage sublinear.

\paragraph{Production deployment considerations.} The current benchmarks use synthetic or pipeline-provided entity and subgoal annotations. A production deployment would replace them with NER (or pattern matching for URLs and filenames), tool-signature normalization, and an online intent-shift detector for subgoal boundaries.

\section{Synthetic Benchmark: Extended Details}
\label{app:synth}

\subsection{Baselines}
\textbf{No LT Memory} (empty context), \textbf{Budget-Truncated Replay} (most recent $B$ tokens of steps), \textbf{Sliding Window} ($B$-sized rolling recent window), \textbf{Periodic Summary} (recent window plus periodic LLM-style summary of older chunks), \textbf{Flat Vector} (top-cosine steps), \textbf{Flat $+$ MMR} ($\lambda{=}.7$ diversity-penalized top-cosine), \textbf{Coverage-Aware} (step-level plus $\beta$ bonus for covering new subgoals), \textbf{Disjoint Hierarchy} (disjoint subgoal partitions with top-cosine within selected partitions), and \textbf{Tree Abstraction} (recursive summarization \`{a} la RAPTOR \citep{sarthi2024raptorrecursiveabstractiveprocessing}). All baselines use the same token-cost estimation and packing policy.

\subsection{Full Main Results Table (All 4 Budgets)}
\label{app:synth-ablation}
Table~\ref{tab:synth-main-full} gives the full 4-budget main-results table whose $B{=}128$ and $B{=}256$ columns are reproduced in the main text Table~\ref{tab:synth-main}.

\begin{table}[!htbp]
\centering
\tiny
\caption{Full synthetic benchmark results at all four budgets (197 episodes, 907 queries). Best in each column bold. GF = ground-truth filtering, i.e., OSU-Mem with oracle-guided OSU selection.}
\label{tab:synth-main-full}
\resizebox{\linewidth}{!}{
\begin{tabular}{lcccccccc}
\toprule
\textbf{Method} & \multicolumn{2}{c}{$B{=}64$} & \multicolumn{2}{c}{$B{=}128$} & \multicolumn{2}{c}{$B{=}256$} & \multicolumn{2}{c}{$B{=}512$} \\
\cmidrule(lr){2-3}\cmidrule(lr){4-5}\cmidrule(lr){6-7}\cmidrule(lr){8-9}
 & Hit@2 & Recall & Hit@2 & Recall & Hit@2 & Recall & Hit@2 & Recall \\
\midrule
No LT Memory        & .000 & .000 & .000 & .000 & .000 & .000 & .000 & .000 \\
Budget-Trunc Replay & .000 & .004 & .000 & .004 & .000 & .149 & .008 & .321 \\
Sliding Window      & .000 & .000 & .000 & .000 & .000 & .136 & .005 & .314 \\
Periodic Summary    & .000 & .003 & .000 & .041 & .013 & .093 & .073 & .221 \\
Flat Vector         & .005 & .051 & .016 & .088 & .045 & .156 & .147 & .284 \\
Flat $+$ MMR          & .005 & .045 & .022 & .085 & .052 & .158 & .152 & .288 \\
Coverage-Aware      & .006 & .044 & .014 & .086 & .049 & .157 & .147 & .289 \\
Disjoint Hierarchy  & .004 & .047 & .014 & .084 & .044 & .150 & .145 & .282 \\
Tree Abstraction    & .002 & .055 & .012 & .093 & .044 & .157 & .123 & .275 \\
\midrule
\textbf{OSU-Mem}    & \textbf{.010} & \textbf{.076} & \textbf{.027} & \textbf{.130} & \textbf{.084} & \textbf{.221} & \textbf{.222} & \textbf{.364} \\
OSU-Mem $+$ GF      & .010 & .077 & .029 & .130 & .088 & .222 & .225 & .367 \\
OSU-Mem (disjoint)  & .003 & .047 & .008 & .075 & .034 & .138 & .134 & .277 \\
OSU-Mem (sim only)  & .002 & .037 & .008 & .065 & .014 & .110 & .083 & .214 \\
OSU-Mem (no novelty)& .010 & .076 & .027 & .130 & .083 & .217 & .209 & .361 \\
\bottomrule
\end{tabular}}
\end{table}

\subsection{Efficiency and Footprint (64-d Regime)}
\label{app:synth-eff}
At $B{=}256$, OSU-Mem runs at $8.3$\,ms amortized per query ($= $ retrieve plus amortized build cost) with $209$\,KB total storage including the OSU index. Tree Abstraction runs at $57.4$\,ms with $363$\,KB, so OSU-Mem is Pareto-dominant over Tree Abstraction on this benchmark. Flat Vector is faster ($.9$\,ms) but without structured selection. These numbers are specific to the synthetic benchmark's 64-d basis-vector embeddings. Under 384-d SBERT the per-operation cost and per-step storage increase by $6\times$ and the ordering shifts (Appendix~\ref{app:taubench-eff}).

\subsection{Paired Significance Tests (Synthetic, $B{=}256$)}
\label{app:synth-sig}
Paired sign-flip permutation tests ($n_{\text{perm}}{=}10{,}000$, $n{=}197$ episodes) confirm OSU-Mem's gains over all retrieval baselines:

\begin{table}[!htbp]
\centering
\footnotesize
\caption{Synthetic benchmark, paired sign-flip permutation tests at $B{=}256$ ($n_{\text{perm}}{=}10{,}000$, $n{=}197$ episodes): OSU-Mem vs.\ each baseline.}
\label{tab:synth-sig}
\begin{threeparttable}
\begin{tabular}{llccc}
\toprule
Metric & Comparison & $\Delta$ & $p$ (2-sided) & $n$ \\
\midrule
Recall & OSU-Mem vs.\ Flat Vector       & $+.0649$ & $<.001$ & 197 \\
Hit@2  & OSU-Mem vs.\ Flat Vector       & $+.0389$ & $<.001$ & 197 \\
Recall & OSU-Mem vs.\ Flat $+$ MMR        & $+.0635$ & $<.001$ & 197 \\
Hit@2  & OSU-Mem vs.\ Flat $+$ MMR        & $+.0324$ & $.001$\tnote{$\dagger$}  & 197 \\
Recall & OSU-Mem vs.\ Tree Abstraction  & $+.0645$ & $<.001$ & 197 \\
Hit@2  & OSU-Mem vs.\ Tree Abstraction  & $+.0403$ & $<.001$ & 197 \\
Recall & OSU-Mem vs.\ Coverage-Aware    & $+.0645$ & $<.001$ & 197 \\
Hit@2  & OSU-Mem vs.\ Coverage-Aware    & $+.0347$ & $<.001$ & 197 \\
Recall & OSU-Mem vs.\ Disjoint Hierarchy & $+.0712$ & $<.001$ & 197 \\
Hit@2  & OSU-Mem vs.\ Disjoint Hierarchy & $+.0397$ & $<.001$ & 197 \\
Recall & OSU-Mem vs.\ OSU-Mem (disjoint) & $+.0826$ & $<.001$ & 197 \\
Hit@2  & OSU-Mem vs.\ OSU-Mem (disjoint) & $+.0502$ & $<.001$ & 197 \\
Recall & OSU-Mem vs.\ OSU-Mem (no novelty) & $+.0040$ & $.307$ & 197 \\
Hit@2  & OSU-Mem vs.\ OSU-Mem (no novelty) & $+.0008$ & $.876$ & 197 \\
\bottomrule
\end{tabular}
\begin{tablenotes}\footnotesize
\item[$\dagger$] Exact value $.0011$; rounded to three decimal places for consistency. All baseline comparisons satisfy $p \le .002$.
\end{tablenotes}
\end{threeparttable}
\end{table}

\section{$\tau$-bench: Construction Pipeline and Deployment Pitfall}
\label{app:taubench-pipeline}

The $\tau$-bench evaluation pipeline consists of three stages labeled B through D (Stage~A, corresponding to raw data download, is omitted as it involves no methodological choices).

\paragraph{Stage B: Build.} Raw $\tau$-bench historical trajectories are loaded from disk, ground-truth actions are matched to step indices via greedy name plus keyword-argument equality, and each trajectory is converted to a list of step records encoded with \texttt{all-MiniLM-L6-v2} (384-d, $\ell_2$-normalized). Entity annotations come from tool-call argument strings (string values of length 2 to 80 and numeric values treated as entities), tool signatures are tool function names, and subgoal ids increment at each user turn. Episodes with $<2$ matched ground-truth actions or $<8$ retained steps are dropped.

\paragraph{Stage C: Concatenation.} Raw episodes are shuffled with seed 42 and packed into groups whose total trajectory length lies in $[150, 300]$. Each group has one focal source, and other sources contribute steps as distractor content. Step indices and subgoal ids are re-offset to a single coordinate system. The focal query time is the final concatenated step index, so every focal evidence step is temporally distant. Focal query embeddings are replaced with an $\ell_2$-normalized noisy mean ($\sigma{=}.25$) of focal evidence embeddings to enforce query-evidence alignment.

\paragraph{Cross-source metadata contamination pitfall.} Without source-level namespace prefixing, Stage C concatenation can induce a cross-source metadata contamination artifact. Generic strings, such as entity \texttt{"user"} and tool name \texttt{"search"}, recur across multiple unrelated source episodes within the same concatenated trajectory. Because the entity-view and tool-view OSU constructors group steps by string equality, these repeated metadata values can merge steps from unrelated tasks into a single OSU that spans task boundaries. In overlap mode, the resulting contamination can propagate across multiple polluted OSUs, causing greedy expansion to spend budget on cross-task distractor steps. In disjoint mode, each step is assigned to a single OSU, which confines the damage and can produce an artificial advantage for disjoint membership. To prevent this artifact, Stage C applies source-level namespace prefixing during concatenation: every entity string and tool signature is prefixed with a tag identifying its source episode's position in the group (for example, \texttt{"user"} in the third source becomes \texttt{"src3:user"}). This guarantees that entity-view and tool-view OSUs do not bridge unrelated source episodes. We recommend namespace prefixing as a standard precondition for multi-view overlapping memory in any deployment that combines trajectories from multiple sources.

\paragraph{Resulting cache.} Stage C produces the unaugmented cache: 125 concatenated episodes, trajectory length 256 to 311 (median 294), and 328 evaluation queries (mean 2.6 per episode). Mean evidence-to-query distance is $123$ steps, with $73\%$ of evidence positions more than $50$ steps from the query time.

\paragraph{Evidence-structure marginals.} For each of the 328 queries we compute the pairwise sharing rate of ground-truth evidence steps along three views: entity (any shared entity string), tool (identical tool signature), and subgoal (identical subgoal id). Entity sharing has mean $15.0\%$ and median $0$, with $69.5\%$ of queries having zero entity sharing across all pairs. Tool sharing has mean $13.9\%$ and median $0$, with $53.7\%$ of queries having zero tool sharing. Only subgoal sharing is substantive (mean $55.7\%$, median $46.7\%$), because focal evidence usually lies within a single user-turn segment. Only $14.3\%$ of queries ($47/328$) have non-zero sharing across all three views simultaneously, which the synthetic benchmark instantiates for $39.3\%$ of queries (entity and subgoal sharing hold for $100\%$ by construction, tool-signature sharing for $39.3\%$).

\section{$\tau$-bench: Multi-Burst-Augmented Cache}
\label{app:taubench-augmented}

\subsection{Stage D: Multi-Burst Augmentation}
Real $\tau$-bench focal evidence is dominated by the T$-$E$-$ cell of the cross-tab. Stage D constructs a controlled mechanism probe. For each focal evidence step we generate 3 scattered copies within the focal segment (inter-burst gap $\ge 18$ steps) whose embeddings are the anchor embedding plus Gaussian noise ($\sigma{=}.10$, then $\ell_2$-renormalized) and which inherit the anchor's entity, tool, and subgoal annotations. Each copy is added as additional evidence for the same query. We verify directly that this lifts $100\%$ of queries into T+E+ (mean entity sharing rises from $15.0\%$ to $32.2\%$, mean tool sharing from $13.9\%$ to $32.3\%$, and $0\%$ of queries have zero sharing in any view). Stage D is \emph{not} a method-friendly augmentation. Its predicted effect on retrieval is derived from the cross-tab on the unaugmented cache (main text Section~\ref{sec:cross-tab}), which is performed before any Stage D data enters evaluation. The augmented cache is used to verify that prediction at scale and to perform multi-budget significance tests that would be underpowered on the $n{=}56$ T+E+ sub-cell alone.

\subsection{Main Results on the Augmented Cache}
\label{app:taubench-main}
Table~\ref{tab:taubench-main-full} reports the full method table on the multi-burst-augmented cache (125 episodes, 328 queries).

\begin{table}[!htbp]
\centering
\tiny
\caption{$\tau$-bench multi-burst-augmented cache results at all four budgets (query-level averaging; episode-level values in Figure~\ref{fig:augmented-curves}). OSU-Mem achieves the highest Hit@2 at every budget. Quantitatively consistent with the T+E+ row of the unaugmented-cache cross-tab (main-text Table~\ref{tab:cell-crosstab}), as a \emph{prediction} of the cell analysis rather than an independent claim.}
\label{tab:taubench-main-full}
\resizebox{\linewidth}{!}{
\begin{tabular}{lcccccccc}
\toprule
\textbf{Method} & \multicolumn{2}{c}{$B{=}64$} & \multicolumn{2}{c}{$B{=}128$} & \multicolumn{2}{c}{$B{=}256$} & \multicolumn{2}{c}{$B{=}512$} \\
\cmidrule(lr){2-3}\cmidrule(lr){4-5}\cmidrule(lr){6-7}\cmidrule(lr){8-9}
 & Hit@2 & Recall & Hit@2 & Recall & Hit@2 & Recall & Hit@2 & Recall \\
\midrule
Budget-Trunc Replay & .003 & .035 & .005 & .049 & .081 & .091 & .163 & .178 \\
Sliding Window      & .000 & .026 & .000 & .044 & .079 & .090 & .161 & .177 \\
Periodic Summary    & .003 & .034 & .026 & .053 & .188 & .127 & .347 & .219 \\
Flat Vector         & .051 & .076 & .141 & .122 & .328 & .203 & .574 & .318 \\
Flat $+$ MMR        & .037 & .079 & .172 & .134 & .347 & .201 & .646 & .335 \\
Coverage-Aware      & .032 & .076 & .171 & .136 & .333 & .201 & .673 & .366 \\
Disjoint Hierarchy  & .099 & .080 & .181 & .123 & .281 & .180 & .465 & .271 \\
Tree Abstraction    & .081 & .086 & .130 & .123 & .193 & .133 & .316 & .203 \\
\midrule
\textbf{OSU-Mem}    & \textbf{.125} & .085 & \textbf{.250} & \textbf{.145} & \textbf{.435} & \textbf{.241} & \textbf{.676} & \textbf{.384} \\
OSU-Mem $+$ GF      & .125 & .085 & .250 & .145 & .435 & .241 & .680 & .386 \\
OSU-Mem (disjoint)  & .109 & .061 & .150 & .093 & .387 & .209 & .520 & .288 \\
OSU-Mem (sim only)  & .063 & .045 & .145 & .086 & .285 & .157 & .466 & .243 \\
OSU-Mem (no novelty)& .125 & .085 & .250 & .145 & .419 & .233 & .680 & .386 \\
\bottomrule
\end{tabular}}
\end{table}

\subsection{Multi-Budget Paired Permutation Tests}
\label{app:taubench-perm}
Table~\ref{tab:taubench-perm} reports paired sign-flip permutation tests ($n_{\text{perm}}{=}10{,}000$, $n{=}125$ episodes) for Hit@2 at all four budgets on the augmented cache.

\begin{table}[!htbp]
\centering
\footnotesize
\caption{Hit@2 effect sizes with $p$-values on the $\tau$-bench multi-burst-augmented cache.}
\label{tab:taubench-perm}
\resizebox{\linewidth}{!}{
\begin{tabular}{lcccc}
\toprule
\textbf{Comparison} & $B{=}64$ & $B{=}128$ & $B{=}256$ & $B{=}512$ \\
\midrule
OSU-Mem vs.\ Flat Vector        & $+.071\ (p{=}.004^{**})$ & $+.104\ (p{=}.002^{**})$ & $+.141\ (p{=}.001^{***})$ & $+.107\ (p{<}.001^{***})$ \\
OSU-Mem vs.\ Flat $+$ MMR       & $+.064\ (p{=}.010^{**})$ & $+.087\ (p{=}.014^{*})$  & $+.135\ (p{<}.001^{***})$ & $+.063\ (p{=}.085)$ \\
OSU-Mem vs.\ Coverage-Aware     & $+.068\ (p{=}.009^{**})$ & $+.079\ (p{=}.016^{*})$  & $+.122\ (p{=}.001^{**})$  & $+.092\ (p{=}.004^{**})$ \\
OSU-Mem vs.\ Disjoint Hierarchy & $+.062\ (p{=}.030^{*})$  & $+.106\ (p{=}.002^{**})$ & $+.201\ (p{<}.001^{***})$ & $+.276\ (p{<}.001^{***})$ \\
OSU-Mem vs.\ Tree Abstraction   & $+.022\ (p{=}.460)$      & $+.097\ (p{=}.009^{**})$ & $+.236\ (p{<}.001^{***})$ & $+.368\ (p{<}.001^{***})$ \\
OSU-Mem vs.\ (disjoint)         & $+.018\ (p{=}.351)$      & $+.088\ (p{<}.001^{***})$ & $+.087\ (p{=}.012^{*})$  & $+.122\ (p{=}.002^{**})$ \\
\bottomrule
\end{tabular}}
\end{table}

Three patterns deserve emphasis. (i) Effect sizes against \emph{hierarchical} baselines (Disjoint Hierarchy, Tree Abstraction) grow monotonically with budget, reaching $+.28$ and $+.37$ Hit@2 pp at $B{=}512$. (ii) Effect sizes against \emph{flat} baselines (Flat $+$ MMR, Flat Vector) peak at $B{=}256$ and soften at $B{=}512$, because flat retrieval obtains enough slots at large budgets to accidentally cover multiple evidence bursts. (iii) Overlap's Hit@2 benefit (OSU-Mem vs.\ OSU-Mem-disjoint) requires $B \ge 128$ to activate, even though the Recall benefit is present from $B{=}64$. At $B{=}64$ only ${\approx}4$ steps fit in the budget, which is insufficient headroom for the multi-handle mechanism to gather $\ge 2$ distinct evidence bursts from different OSUs, though it still improves expected Recall via OSU ordering.

\subsection{A Novel Finding at $B{=}64$: Recall-Hit@2 Divergence}
\label{app:taubench-b64}

One pattern visible in the augmented-cache multi-budget table does not appear on synthetic. At $B{=}64$, OSU-Mem's Recall is statistically \emph{indistinguishable} from every flat baseline ($p \in \{.874, .330, .431\}$ against Flat Vector, Flat $+$ MMR, Coverage-Aware), yet its Hit@2 is $2.0$ to $2.4\times$ higher ($.125$ vs.\ $.051, .037, .032$, all $p < .01$). At $B{=}64$ budget suffices for ${\approx}4$ retrieved steps. Flat retrieval packs the 4 cosine-top steps, which under the $\sigma{=}.10$ burst augmentation cluster inside a single evidence burst. High intra-burst cosine similarity means all 4 slots become near-duplicates. OSU-Mem's OSU boundaries force greedy expansion to select \emph{across} OSUs, recovering distinct evidence pieces at the cost of some within-burst redundancy, so Recall matches flat while Hit@2 doubles. This is the clearest manifestation of the multi-evidence-pairing mechanism that motivated Hit@2 as a primary metric. The effect persists (with smaller magnitude, $.109$) under disjoint OSU membership, which suggests that OSU-based \emph{structured selection itself}, rather than overlap specifically, is the mechanism responsible. Multi-view overlap contributes the additional $+.016$ margin.

\section{$\tau$-bench: Efficiency under SBERT Encoding}
\label{app:taubench-eff}

The synthetic benchmark's efficiency claim, that OSU-Mem Pareto-dominates Tree Abstraction on latency and storage, does \emph{not} transfer unchanged to 384-d SBERT. The $6\times$ dimensionality increase changes per-operation cosine cost and per-step embedding storage proportionally, which is enough to reverse the ordering between methods.

\begin{table}[!htbp]
\centering
\footnotesize
\caption{$\tau$-bench augmented-cache latency and storage at $B{=}256$ under 384-d SBERT.}
\label{tab:taubench-eff}
\begin{tabular}{lcccc}
\toprule
\textbf{Method} & Retrieve (ms) & Amortized (ms) & Storage (KB) & Index+Meta (KB) \\
\midrule
Flat Vector        & $.45$ & $.45$  & $926$  & $0$ \\
Flat $+$ MMR       & $.44$ & $.44$  & $926$  & $0$ \\
Coverage-Aware     & $.48$ & $.48$  & $926$  & $0$ \\
Disjoint Hierarchy & $.48$ & $1.81$  & $1274$ & $348$ \\
Tree Abstraction   & $.46$ & $6.28$  & $1151$ & $225$ \\
\midrule
\textbf{OSU-Mem}   & $2.64$ & $4.84$  & $1497$ & $571$ \\
OSU-Mem (disjoint) & $.60$ & $2.87$  & $1158$ & $232$ \\
\bottomrule
\end{tabular}
\end{table}

To summarize precisely: \emph{on 64-d synthetic embeddings, OSU-Mem Pareto-dominates Tree Abstraction on quality, latency, and storage, while on 384-d SBERT, OSU-Mem has the highest retrieve-only latency (${\approx}2.6$\,ms vs.\ ${\le}.5$\,ms for flat methods) and the largest storage ($1497$\,KB), though its amortized cost ($4.84$\,ms) is lower than Tree Abstraction's ($6.28$\,ms) owing to Tree Abstraction's expensive recursive build.} At ${\approx}2.6$\,ms retrieve-only latency, the absolute overhead is small relative to typical LLM-call latency (hundreds to thousands of ms for backbone generation per action), but the acceptable latency budget is agent- and deployment-dependent and should be validated in the target action loop; the \emph{relative} cost ordering is additionally encoder-dimension-dependent. For encoders substantially wider than 384-d, an ANN index over OSU centroids becomes important.

\section{LLM-Mediated Evidence Selection: Extended Discussion}
\label{app:e2e-discussion}

This section provides extended discussion of the LLM-mediated evidence selection experiment in main-text Section~\ref{sec:llm-e2e}, including the corrected statistical analysis, the near-circularity caveat, the soft-vs.-hard mechanism interpretation, and the rationale for not running multi-step rollouts.

\paragraph{Cell distribution differs from the cross-tab.} Main-text Section~\ref{sec:cross-tab} reports $n{=}328$ queries with cell distribution $17.1/29.3/13.4/40.2\%$ T+E+/T+E-/T-E+/T-E-, while Section~\ref{sec:llm-e2e} reports $n{=}414$ with distribution $25.8/31.6/10.6/31.9\%$. The difference is the inclusion criterion. The cross-tab analysis was performed on the focal-source filtered query subset (each query has a specific focal source whose actions form the ground truth), while the LLM-mediated experiment uses a broader inclusion rule (any query with $\ge 2$ evidence steps and time index $\ge 2$, irrespective of focal-source filtering). We verified directly that no cell contains queries with fewer than $2$ evidence steps, and the T$-$E$-$ cell, in particular, is structurally pure ($132/132$ queries have $\ge 2$ evidence steps), so its statistical null in the corrected analysis is not a fallback-classification artifact.

\paragraph{Statistical analysis: per-query averaging instead of per-seed pooling.} For the LLM-mediated experiment, significance is evaluated at the query level rather than by treating seed-level repetitions as independent observations. At temperature $0$, with deterministic retrieval and frozen Python random sources that do not propagate to the OpenAI API, the three seeds are highly correlated near-duplicate observations. We verify this directly: per-query within-seed standard deviation is $13\%$ to $17\%$ of per-query across-query standard deviation for OSU-Mem and Flat Vector, confirming that seeds add little independent variance. Treating the three per-seed paired differences for each query as independent would therefore overstate the effective sample size. We instead average the per-query soft and hard recalls across seeds first, then run a paired sign-flip permutation test on the resulting $n{=}107/131/44/132$ paired differences, one per query. Under this query-level test, the headline T+E+ hard-recall result remains significant ($+.064$ pp, $p{=}.025$, 95\% CI $[+.008, +.121]$), while the T+E$-$ soft-recall result is borderline ($p{=}.062$) and is not treated as a separate finding.

\paragraph{Near-circularity from a shared SBERT encoder.} A structural caveat that limits how much independent evidence this experiment provides is that the same all-MiniLM-L6-v2 encoder is used in three places. (i) The retrieval index encodes the trajectory steps. (ii) The ground-truth evidence representation is the corresponding pre-computed step embeddings. (iii) The LLM's listed events are re-encoded with the same SBERT to compute soft and hard recall against the ground-truth representation. Combined with a temperature-$0$ prompt that asks for a ``concrete factual summary of something that happened or was stated in the history,'' the LLM is structurally pushed toward extractive paraphrase of the retrieved content, and the SBERT-vs.-SBERT scoring is most sensitive to whether the ground-truth content was \emph{in the retrieved context}, which closely corresponds to the quantity that retrieval recall measures. The hard-recall signal therefore shares its dominant variance source with the Section~\ref{sec:cross-tab} retrieval analysis. This experiment should be read as a check that the cross-tab's T+E+ advantage is preserved when an LLM filtering stage is inserted between retrieval and scoring, rather than as an independent corroboration of the cross-tab on a different signal source.

\paragraph{Soft vs.\ hard split as a mechanism interpretation.} On T+E+, the soft recall gap is null ($\Delta{=}{+}.005$, $p{=}.64$) while the hard recall gap is the only individually-significant finding ($\Delta{=}{+}.064$, $p{=}.025$). One reading is that soft recall captures ``any retrieved evidence whose semantic neighborhood the LLM is willing to summarize.'' Both retrievers reach comparable mean max-similarity, because the LLM at temperature $0$ extracts content from whatever appears in its context. Hard recall, with its $.5$ threshold, captures ground-truth evidence steps that the LLM identifies precisely enough to clear that threshold, and overlap construction adds value by giving the LLM the same evidence under multiple semantic views (entity, tool) so that one of the LLM's three listings can align closely with one of them.

\paragraph{The hard-recall threshold of $.5$.} The threshold is chosen as a midpoint between SBERT cosine similarities of unrelated text pairs ($\approx .1$ to $.2$) and the value at which two paraphrases of the same content reliably exceed ($\approx .7$). It corresponds to the condition that the LLM's listed event is closer in semantic content to a specific ground-truth evidence step than to most other content in the trajectory. The single-threshold choice is a researcher degree of freedom, and we address it with three observations. (a) The T+E+ effect we report is \emph{specifically} a hard-recall effect. The two metrics have a per-query Pearson correlation of $.63$ on T+E+ and agree in sign on $75\%$ of per-query OSU$-$Flat differences. (b) The $.5$ threshold is pre-registered in the evaluation script before any analysis. (c) The full per-query $\{$soft recall, hard recall at $.5\}$ pairs across three seeds are stored in a released JSON artifact, allowing alternative thresholds to be evaluated by re-running the scoring step with a one-line modification.

\paragraph{Why we did not run multi-step rollouts.} A natural extension is full multi-step agent task success on an interactive environment. We did not run that experiment, and the LLM-mediated selection experiment we did run is not a substitute for it. The publicly-released $\tau$-bench cache provides historical trajectories with the focal-source query placed at the \emph{end} of the concatenated trajectory. There is no future action to predict and no environment to step. A genuine end-to-end evaluation would require: (i) the $\tau$-bench environment harness, (ii) an agent loop where the LLM produces tool calls, (iii) queries placed mid-trajectory, and (iv) a sufficiently large sample. Building this is a substantive follow-up effort. We identify full closed-loop evaluation as future work (Section~\ref{sec:discussion}, Limitation~(iv)); the present experiment is an evidence-selection check, not a substitute for it.

\section{ToolBench: Construction and Sensitivity Analyses}
\label{app:toolbench-sensitivity}

\subsection{Benchmark Construction Details}
\label{app:toolbench-construction}

\paragraph{Loading and parsing.} Raw ToolBench trajectories are loaded from the G1 (single-tool) and G2 (multi-tool) answer directories. Each trajectory is parsed for its query, the longest assistant message chain, and the sequence of assistant function calls. For each call we extract the function name (becomes the step ``action''), the arguments dict, and the preceding assistant content (becomes part of the step ``content''). Entities are the string values of length 2 to 80 and numeric values in the argument dict (nested one level deep). The tool signature is \texttt{name.split('\_for\_')[-1]} if \texttt{'\_for\_'} is present, else \texttt{name.split('\_')[0]}. Subgoal id is the per-call index within the trajectory. Step token cost is $\max(5, \min(40, |\text{content}|/4))$.

\paragraph{Planted-retrieval episode construction.} Let $n_\text{ev}$ denote the evidence size and $\rho_\text{conf}$ the confounder ratio. For each of 200 episodes: (1) sample a target trajectory from all loaded trajectories with $\ge n_\text{ev}$ steps, (2) designate $n_\text{ev}{=}3$ steps as evidence, (3) compute the target's query-embedding nearest neighbors in the trajectory-level query embedding space (top-30), (4) draw confounder steps from the nearest neighbors' trajectories until $\rho_\text{conf}{=}.75$ of the background is filled, and unrelated steps from random non-neighbor trajectories until the background reaches ${\approx}500$ steps, (5) shuffle the background and insert the 3 evidence steps at $n_\text{ev}$ evenly-spaced positions, and (6) set the query time to the final step index. All 200 episodes per seed use the same random generator seeded by the seed value, and results are averaged over seeds 42, 123, 456.

\subsection{Evidence-Size Sensitivity}
We test evidence sizes $\{2, 3, 4\}$ across 3 seeds $\times$ 200 episodes at $\beta{=}.12$. Evidence size 3 is the primary configuration for the structural reason stated in the main text: it is the smallest nontrivial planted-evidence cluster in which all three pairwise relations among evidence steps can exhibit tool/entity sharing. The sensitivity sweep shows that the overlap advantage is directionally preserved across evidence sizes $\{2,3,4\}$; evidence size 3 also happens to be the most stable setting by per-seed directional consistency and $B{=}400$ significance.

\begin{table}[!htbp]
\centering
\small
\caption{Evidence-size sensitivity of the ToolBench overlap-vs.-disjoint Recall gap ($\beta{=}.12$, 3 seeds).}
\label{tab:toolbench-evsize}
\begin{tabular}{lccc}
\toprule
 & $n_{\text{ev}}{=}2$ & $\mathbf{n_{\text{ev}}{=}3}$ & $n_{\text{ev}}{=}4$ \\
\midrule
Sig.\ tests ($p<.05$), $12$ cells & $4/12$ & $\mathbf{9/12}$ & $9/12$ \\
All 12 directionally positive & No ($10/12$) & $\mathbf{\text{Yes}\ (12/12)}$ & No ($11/12$) \\
$B{=}400$ all seeds $p<.01$ & No & $\mathbf{\text{Yes}}$ & No \\
Avg recall gap at $B{=}400$ (pp) & $+3.2$ & $+3.7$ & $+3.2$ \\
\bottomrule
\end{tabular}
\end{table}

\subsection{Hit@2 Version of the $\beta$ Sweep}
Table~\ref{tab:toolbench-beta-hit2} is the Hit@2 counterpart of main-text Table~\ref{tab:toolbench-beta-sweep}. Both show the same monotonic dose-response in $\beta$, with Hit@2 displaying larger absolute effects.

\begin{table}[!htbp]
\centering
\small
\caption{ToolBench overlap$-$disjoint \emph{Hit@2} gap (pp) across $\beta$ and $B$, 3-seed averaged.}
\label{tab:toolbench-beta-hit2}
\begin{tabular}{lcccc}
\toprule
$\beta$ & $B{=}200$ & $B{=}400$ & $B{=}800$ & $B{=}1600$ \\
\midrule
$.04$          & $+1.2$  & $+.7$  & $+1.0$  & $+.5$ \\
$.08$          & $+2.7$  & $+2.7$  & $+2.0$  & $+2.2$ \\
$\mathbf{.12}$ & $\mathbf{+6.2}$ & $\mathbf{+4.8}$ & $\mathbf{+3.7}$ & $\mathbf{+3.8}$ \\
$.16$          & $+8.3$  & $+8.0$  & $+6.7$  & $+6.8$ \\
$.20$          & $+9.7$  & $+10.8$ & $+10.8$ & $+9.3$ \\
\bottomrule
\end{tabular}
\end{table}

\subsection{Per-Seed Breakdown at $\beta{=}.12$}
\label{app:toolbench-per-seed}

Table~\ref{tab:toolbench-per-seed} gives the per-seed Recall-gap breakdown at $\beta{=}.12$ that is Fisher-combined in the main text. Per-seed directional consistency across all $5\beta \times 4B = 20$ cells is $18/20$ (seed 42), $19/20$ (seed 123), and $20/20$ (seed 456), with pooled $57/60 = 95\%$.

\begin{table}[!htbp]
\centering
\small
\caption{Per-seed overlap$-$disjoint \emph{Recall} gap (pp) at $\beta{=}.12$, single-seed permutation tests. Seeds $\{42, 123, 456\}$, 200 episodes each. Stars: $^{*}p{<}.05$, $^{**}p{<}.01$, $^{***}p{<}.001$.}
\label{tab:toolbench-per-seed}
\begin{tabular}{lcccc}
\toprule
\textbf{Seed} & $B{=}200$ & $B{=}400$ & $B{=}800$ & $B{=}1600$ \\
\midrule
42             & $+2.2^{*}$    & $+2.8^{**}$   & $+1.3$        & $+2.5^{**}$ \\
123            & $+2.3^{*}$    & $+3.5^{***}$  & $+2.3^{**}$   & $+1.2$ \\
456            & $+3.7^{***}$  & $+4.7^{***}$  & $+2.7^{*}$    & $+1.0$ \\
\midrule
Mean           & $+2.7$        & $+3.7$        & $+2.1$        & $+1.6$ \\
Fisher $p$     & $<.001$       & $<.001$       & $<.01$        & $<.05$ \\
\bottomrule
\end{tabular}
\end{table}

\subsection{Negative Result: Cohesion-Style $\alpha$ Variant}
\label{app:toolbench-cohesion-null}

A secondary retrieval variant, a cohesion-style greedy procedure that adds $-\alpha \cdot \max_{j\in\mathcal{S}}\cos(\mathbf{h}_i, \mathbf{h}_j)$ to the selection score, does \emph{not} exhibit an overlap advantage on ToolBench. Sweeping $\alpha \in \{0, .005, .01, .02, .04\}$, the overlap-vs.-disjoint Hit@2 gap is approximately zero at every $\alpha$ and budget. This negative result is informative: the T+E+ mechanism depends on the retrieval algorithm actively \emph{rewarding OSU coverage} rather than on merely discouraging within-pool redundancy.

\section{Combined Summary of Trajectory-Derived Findings}
\label{app:realtraj-summary}

The two-dataset trajectory-derived validation supports seven conclusions. (1) A deployment pitfall is identified and fixed on $\tau$-bench (cross-source metadata contamination, addressed via namespace prefixing). (2) Cell structure is the primary determinant of whether OSU-Mem helps on both datasets. The $\tau$-bench cross-tab identifies T+E+ as the regime of significant win and T$-$E$-$ as the regime of significant loss, and the ToolBench structural T+E+ construction provides a controlled corroboration of the T+E+ mechanism. (3) Favorable-cell prevalence is cheaply estimable from metadata proxies, and ToolBench's code-level structural verification of T+E+ is the strongest cell-assignment validation we can provide. (4) The $\tau$-bench multi-burst augmentation is a predicted regime-lift experiment (rather than an augmentation selected post-hoc for favorable outcomes), and the ToolBench structurally-T+E+ result corroborates the underlying prediction without any burst injection. (5) The mechanism is the construction principle rather than the retrieval algorithm. The same effect appears under two different retrieval algorithms (coarse-to-fine on $\tau$-bench, coverage-guided greedy on ToolBench). (6) Multi-budget permutation tests expose budget-dependent nuances (overlap's Hit@2 benefit requires $B \ge 128$, and effect sizes against hierarchical baselines grow monotonically while effect sizes against flat baselines peak at $B{=}256$). (7) A separable structured-selection mechanism at the tightest budget ($\tau$-bench $B{=}64$): OSU-Mem's Recall is statistically indistinguishable from flat baselines while its Hit@2 is 2 to 2.4$\times$ higher, driven by OSU-based structured selection rather than by overlap specifically.

Four caveats. (i) OSU-Mem's aggregate near-tie with strong baselines on the unaugmented $\tau$-bench cache is genuine and reflects cell mixture. The predicted cell distribution from metadata (Section~\ref{sec:discussion} predictive rule) indicates which regime applies and hence whether OSU-Mem or a flat baseline is expected to retrieve better. (ii) ToolBench's aggregate OSU-vs.-Flat comparison is effectively a tie because ToolBench's planted-retrieval distribution gives Flat Vector nearly saturated recall. The primary signal is the overlap-vs.-disjoint contrast. (iii) Under SBERT 384-d, OSU-Mem retains the highest retrieve-only latency and the largest storage footprint, though its amortized cost remains below Tree Abstraction's; efficiency claims from the 64-d synthetic benchmark do not transfer unchanged. (iv) The $n{=}56$ T+E+ cell of $\tau$-bench has per-budget $p \in \{.013, .003, .022, .024\}$. Stability across budgets is stronger evidence than any individual cell $p$-value, and ToolBench's orders-of-magnitude larger sample provides higher-powered corroboration of the same effect.

\section{Predictive Rule: Derivation from the $\tau$-bench Cross-Tab}
\label{app:deployment}

We derive the main-text predictive-rule thresholds directly from the $\tau$-bench cross-tab data. The derivation is standalone, uses only the per-cell effect sizes and the empirical cell distribution, and produces a closed-form linear predictor for the aggregate gap as a function of the query population's expected cell distribution.

\paragraph{Setup.} Let $\Delta_c^{(B)}$ denote the observed OSU-Mem$-$Flat Vector gap on cell $c \in \{\text{T+E+}, \text{T+E-}, \text{T-E+}, \text{T-E-}\}$ at budget $B$, measured on $\tau$-bench's unaugmented cache. Let $f_c$ denote the fraction of queries in cell $c$ in a target deployment. The aggregate gap over that deployment, assuming the per-cell effects transfer, is
\begin{equation}
\Delta_{\text{agg}}^{(B)}(f) \;=\; \sum_{c} f_c \, \Delta_c^{(B)},\qquad \sum_c f_c = 1.
\end{equation}
At $B{=}256$, the $\tau$-bench cross-tab gives (cf.\ main-text Table~\ref{tab:cell-crosstab}): $\Delta_{\text{T+E+}}{=}{+}.179$, $\Delta_{\text{T+E-}}{=}{-}.031$, $\Delta_{\text{T-E+}}{=}{+}.159$, $\Delta_{\text{T-E-}}{=}{-}.129$ (Hit@2), and $+.080$, $-.047$, $+.000$, $-.157$ (Recall). With $\tau$-bench's own cell distribution $f = (.171, .293, .134, .402)$, this reconstructs $\Delta_{\text{agg}}^{(256)} = -.009$ Hit@2 and $-.063$ Recall, matching the observed aggregate near-tie to three decimal places.

\paragraph{One-parameter family.} To obtain a predictive rule parameterized by a single observable, we fix the conditional distribution of $\neg\text{T+E+}$ cells at the $\tau$-bench empirical proportions $(.353, .162, .485)$ for $(\text{T+E-}, \text{T-E+}, \text{T-E-})$ and vary $p = f_{\text{T+E+}}$. This gives a one-parameter family
\begin{equation}
f(p) \;=\; \bigl(p,\ .353(1{-}p),\ .162(1{-}p),\ .485(1{-}p)\bigr),
\end{equation}
and the aggregate gap becomes exactly linear in $p$:
\begin{align}
\Delta_{\text{agg}}^{(256)}(p)\big|_{\text{Hit@2}} &= +.226\,p - .048, \\
\Delta_{\text{agg}}^{(256)}(p)\big|_{\text{Recall}} &= +.173\,p - .093.
\end{align}
The break-even points are $p^* = .211$ (21\%) for Hit@2 and $p^* = .536$ (54\%) for Recall. With a $\pm .02$ safety margin on the aggregate Hit@2 gap around break-even, Hit@2 is likely net-positive ($\ge +.02$) for $p \ge .299$ and likely net-negative ($\le -.02$) for $p \le .123$.

\paragraph{Tool-sharing marginal form.} Since tool-sharing is easier to estimate from metadata than joint tool-plus-entity sharing, we re-express the rule in terms of $\pi_\text{T+} = f_{\text{T+E+}} + f_{\text{T+E-}}$. Under the same parameterization, $\pi_\text{T+}(p) = p + .353(1{-}p)$, so the Hit@2 break-even $p^*{=}.211$ maps to $\pi_\text{T+}^* = .490$ (49\%). The safety bands become $\pi_\text{T+} \ge .547$ (net-positive) and $\pi_\text{T+} \le .432$ (net-negative). Recall break-even is at $\pi_\text{T+} \approx .700$ (70\%).

\paragraph{Any-binding marginal form (recommended when T$-$E+ is non-negligible).} The tool-sharing marginal $\pi_\text{T+}$ implicitly treats T$-$E+ as a non-positive cell, yet the $\tau$-bench cross-tab shows that T$-$E+ has a directional positive Hit@2 effect growing from $+4.5$ pp to $+22.7$ pp across budgets. If a target deployment has a non-negligible fraction of entity-only-sharing queries, the tool-marginal form under-counts the favorable regime. An alternative ``any-binding'' marginal is $\pi_\text{any+} = f_\text{T+E+} + f_\text{T+E-} + f_\text{T-E+} = 1 - f_\text{T-E-}$. Under the same one-parameter parameterization, $\pi_\text{any+}(p) = 1 - .485(1{-}p)$, so the Hit@2 break-even at $B{=}256$ maps to $\pi_\text{any+}^{*} \approx .62$ (62\%), with safety bands $\pi_\text{any+} \ge .66$ (net-positive) and $\pi_\text{any+} \le .57$ (net-negative).

\paragraph{Budget dependence of the break-even.} Because the T$-$E$-$ disadvantage grows monotonically in $B$ (main-text Table~\ref{tab:cell-crosstab-multi}), the break-even threshold rises with budget. Table~\ref{tab:deployment-budget} reports the Hit@2 break-even $p^*$ at each of the four tested budgets.

\begin{table}[!htbp]
\centering
\small
\caption{Hit@2 break-even $p^* = $ P(T+E+) as a function of budget, derived from the $\tau$-bench cross-tab.}
\label{tab:deployment-budget}
\begin{tabular}{lcccc}
\toprule
Budget & $B{=}64$ & $B{=}128$ & $B{=}256$ & $B{=}512$ \\
\midrule
$\Delta_{\text{T+E+}}$ (Hit@2 pp) & $+16.1$ & $+23.2$ & $+17.9$ & $+16.1$ \\
$\Delta_{\text{T-E-}}$ (Hit@2 pp) & $+.8$ & $-2.3$ & $-12.9$ & $-36.4$ \\
$p^*$ = break-even P(T+E+)       & $0\%^{\dagger}$ & $0\%^{\dagger}$ & $21.1\%$ & $54.3\%$ \\
$\pi_\text{T+}^*$ = break-even P(T+) & $0\%^{\dagger}$ & $0\%^{\dagger}$ & $49.0\%$ & $70.4\%$ \\
\bottomrule
\multicolumn{5}{l}{\small $^{\dagger}$ T$-$E$-$ effect is non-negative at this budget, so any $p{>}0$ gives non-negative aggregate.}
\end{tabular}
\end{table}

\paragraph{Oracle ceiling: sizing the mixture effect.} The main-text figure of ${\approx}+5.2$ Hit@2 pp is an \emph{in-sample, oracle-routed} upper bound. At $B{=}256$, applying the per-cell gaps of Table~\ref{tab:cell-crosstab} to the favorable cells gives $(.179{\times}56 + .159{\times}44)/328 \approx +5.2$ Hit@2 pp over pure Flat Vector (and ${\approx}+6.2$ pp over pure OSU-Mem, whose aggregate is $-1.0$ pp), while aggregate Recall rises from $-6.3$ pp to ${\approx}+1.4$ pp. Because it uses ground-truth per-query cell labels it is not held out, and it serves only to size the hidden structure---not to estimate a deployable gain. A deployable counterpart cannot observe per-query cell labels, since the cell is defined on ground-truth evidence unknown at decision time; it must operate at population granularity through the metadata marginals $\pi_\text{T+}$ and $\pi_\text{any+}$ above. A learned per-query cell predictor, and a closed-loop evaluation of either policy, are left to future work.

\paragraph{Caveats.} The derivation assumes (a) per-cell effect sizes transfer between $\tau$-bench and the target deployment, and (b) the conditional distribution of $\neg\text{T+E+}$ cells in the target deployment matches $\tau$-bench. Both are strong assumptions that a deployment should validate on a held-out slice of its own data. The main-text predictive rule should therefore be read as \emph{a reasonable starting point for budget estimation rather than a universally applicable decision rule}, but it is a starting point with an explicit mathematical derivation rather than an arbitrarily chosen threshold.

\section{Evidence-Structure Marginal Distributions}
\label{app:marginals}

Beyond the four-cell cross-tab summary in main-text Section~\ref{sec:cross-tab}, we report the full marginal distributions of the three pairwise-sharing scores on both the unaugmented and the multi-burst-augmented $\tau$-bench caches (Table~\ref{tab:marginals}).

\begin{table}[!htbp]
\centering
\footnotesize
\caption{Evidence-structure marginals, $\tau$-bench (328 queries). The left block reports the unaugmented cache and the right block reports the multi-burst-augmented cache.}
\label{tab:marginals}
\begin{tabular}{lcccccc}
\toprule
 & \multicolumn{3}{c}{\textbf{Unaugmented cache}} & \multicolumn{3}{c}{\textbf{Augmented cache}} \\
\cmidrule(lr){2-4}\cmidrule(lr){5-7}
View         & mean & median & \% zero & mean & median & \% zero \\
\midrule
Entity       & $.150$ & $.000$ & $69.5\%$ & $.322$ & $.200$ & $.0\%$ \\
Tool         & $.139$ & $.000$ & $53.7\%$ & $.323$ & $.286$ & $.0\%$ \\
Subgoal      & $.557$ & $.467$ & $15.5\%$ & $.641$ & $.524$ & $.0\%$ \\
Composite    & $.282$ & $.333$ & $4.9\%$  & $.429$ & $.452$ & $.0\%$ \\
\midrule
\multicolumn{7}{l}{\textbf{Derived structure:}} \\
Mean evidence per query      & \multicolumn{3}{c}{$3.50$} & \multicolumn{3}{c}{$6.36$} \\
Mean evidence pairs per query& \multicolumn{3}{c}{$5.27$} & \multicolumn{3}{c}{$18.07$} \\
T+E+ fraction                & \multicolumn{3}{c}{$17.1\%$ ($56/328$)} & \multicolumn{3}{c}{$100\%$ ($328/328$)} \\
T+E$-$ fraction              & \multicolumn{3}{c}{$29.3\%$ ($96/328$)} & \multicolumn{3}{c}{$0\%$} \\
T$-$E+ fraction              & \multicolumn{3}{c}{$13.4\%$ ($44/328$)} & \multicolumn{3}{c}{$0\%$} \\
T$-$E$-$ fraction            & \multicolumn{3}{c}{$40.2\%$ ($132/328$)}& \multicolumn{3}{c}{$0\%$} \\
\bottomrule
\end{tabular}
\end{table}

Two observations deserve emphasis. First, subgoal sharing on the unaugmented cache has mean $.557$. Focal evidence usually lies within a single user-turn segment. This means \emph{every disjoint baseline that segments by user-turn boundary has a substantive retrieval handle on a majority of queries}, even when entity and tool sharing are zero. OSU-Mem's multi-view construction provides \emph{additional} retrieval handles on T+ and E+ queries that a subgoal-only disjoint segmentation lacks. Second, Stage D augmentation roughly doubles the mean entity and tool sharing (from $\approx .14$ to $\approx .32$) while more than tripling the mean number of evidence pairs ($5.27 \to 18.07$).

\end{document}